\documentclass[final,3p,times,twocolumn]{elsarticle}
\usepackage{latexsym}
\usepackage{amssymb}
\usepackage{amsthm}
\usepackage[cmex10]{amsmath}

\usepackage{graphicx}
\usepackage{float}
\usepackage{subfig}
\usepackage[]{placeins}
\usepackage[]{algorithm}
\usepackage{algpseudocode}
\usepackage{multirow}
\usepackage[subnum]{cases}
\usepackage{tabularx}

\usepackage[]{hyperref} 
\usepackage[]{amsrefs} 

\newcommand\noi{\noindent}

\renewcommand{\algorithmicforall}{\textbf{for each}}

\makeatletter
\renewcommand\appendix{\par
  \setcounter{section}{0}%
  \setcounter{subsection}{0}%
  \setcounter{equation}{0}%
  \setcounter{table}{0}
  \setcounter{figure}{0}
  \gdef\theequation{\@Alph\c@section.\arabic{equation}}%
  \gdef\thefigure{\@Alph\c@section.\arabic{figure}}%
  \gdef\thetable{\@Alph\c@section.\arabic{table}}%
  \gdef\thesection{\Alph{section}}%
  \@addtoreset{equation}{section}%
  \@addtoreset{table}{section}
  \@addtoreset{figure}{section}
}
\makeatother

\begin{document}
\begin{frontmatter}
\title{Mining frequent items in unstructured P2P networks}
\author [unile] {Massimo~Cafaro\corref{cor1}}
\ead{massimo.cafaro@unisalento.it}
\cortext[cor1]{Corresponding author}
\author [unile] {Italo Epicoco}
\ead{italo.epicoco@unisalento.it}
\author [unile] {Marco Pulimeno}
\ead{marco.pulimeno@unisalento.it}
\address[unile]{University of Salento, Lecce, Italy}

\begin{abstract} 
Large scale decentralized systems, such as P2P, sensor or IoT device networks are becoming increasingly common, and require robust  protocols to address the challenges posed by the distribution of data and the large number of peers belonging to the network. In this paper, we deal with the problem of mining frequent items in unstructured P2P networks. This problem, of practical importance, has many useful applications. We design P2PSS, a fully decentralized, gossip--based protocol for frequent items discovery, leveraging the Space-Saving algorithm. We formally prove the correctness and theoretical error bound. Extensive experimental results clearly show that P2PSS provides very good accuracy and scalability, also in the presence of highly dynamic P2P networks with churning. To the best of our knowledge, this is the first gossip--based distributed algorithm providing strong theoretical guarantees for both the Approximate Frequent Items Problem in Unstructured P2P Networks and for the frequency estimation of discovered frequent items.
\end{abstract}


\newtheorem{prob}{Problem}
\newtheorem{thm}{Theorem}
\newtheorem{lem}[thm]{Lemma}
\newdefinition{rmk}{Remark}
\newproof{pf}{Proof}
\newtheorem{prop}[thm]{Proposition}
\newtheorem*{cor}{Corollary}
\newdefinition{defn}{Definition}
\newtheorem{conj}{Conjecture}
\newtheorem{exmp}{Example}
\newtheorem{case}{Case}

\end{frontmatter}


\section*{Declaration of interest}

Declarations of interest: none.

\section{Introduction}
\label{intro}

Large scale decentralized systems, such as P2P (Peer to Peer), sensor or IoT (Internet of Things) device networks are becoming increasingly common. As an example, P2P based systems underlie popular sharing platforms allowing data exchange among a large number of users. However, dissemination and delivery of valuable data and information is complicated by the distributed nature of the network. The lack of a central authority in charge of administration forces the need for fully decentralized protocols in which the peers interact and collaborate towards a common goal.

In the case of structured P2P networks, the underlying topology may be exploited in the design of a distributed protocol. However, for unstructured networks, the lack of a specific topology must be also taken into account. A possibility, commonly found in many protocols, is to impose a topology: these protocols rely on the construction of a spanning tree, which is then used for information dissemination. A popular alternative is the use of gossip--based communication mechanisms. Informally, a gossip--based protocol can be though of as a sequence of rounds in which each peer randomly selects one or more peers, exchanges its local state information with the selected peers and updates its local state by using the received information. 

Owing to the randomized choices made by the peers in each round of the distributed computation, it may appear somewhat surprising that gossip--based protocols can provide a fast and accurate solution to the problem of providing a consistent global view of the information locally stored at each peer. 

 In this paper, we deal with the problem of mining frequent items in unstructured P2P networks. Mining of frequent items (also known as heavy hitters) is a problem of fundamental importance, both from a theoretical and practical perspective, as witnessed by the considerable attention and recognition received, which led to a huge number of related publications. Different scientific communities refer to the problem as \textit{market basket analysis} \cite{Brin}, \textit{hot list analysis} \cite{Gibbons} and \textit{iceberg query} \cite{Fang98computingiceberg}, \cite{Beyer99bottom-upcomputation}. 
 
 Among the many possible applications, consider a large P2P network such as BitTorrent and the need to collect useful statistics on the service, such as the most frequently accessed files. The relevant  information is distributed amongst the peers, therefore applications that need a global view of such information/statistics encounter particular difficulties to operate, and a distributed algorithm is required to solve the problem. The optimization of cache performance in distributed storage systems and the performance improvement of distributed information retrieval in search engines obviously require the knowledge of the most frequently accessed data and, respectively, metadata. Distributed frequent items algorithm can also help detecting Internet worms or DDoS (Distributed Denial of Service) attacks to a network, by respectively tracking frequently recurring bit strings, or frequently accessed web servers, and reporting frequencies above a specified threshold \cite{MENG201629}. The problem of detecting \textit{superspreaders}, which are sources that connect to a large number of distinct destinations, is also useful in P2P networks, where it could be used to find peers that talk to a lot of other peers without keeping per-peer information as in traditional approaches. 
 
 Other possible applications concern frequent queries, globally across the whole network:
 
  \begin{itemize}
  \item Popular products. The input may be the page views of products on Amazon yesterday; heavy hitters are then the most frequently viewed products;
  \item Popular search queries. The input may consist of all of the searches on Google yesterday; heavy hitters are then searches made most often;
  \item TCP flows. The input may be the data packets passing through a network switch, each annotated with a source-destination pair of IP addresses. The heavy hitters are then the flows that are sending the most traffic.
\end{itemize}
 
 We recall here other applications, including network traffic analysis \cite{DemaineLM02},  \cite{Estan}, \cite{Pan}, analysis of web logs \cite{Charikar}, Computational and theoretical Linguistics \cite{CICLing}.

 The problem can be solved by designating one of the peers as a central manager, and letting each peer communicate its local information to the manager peer. Once the whole dataset has been obtained, the manager peer solves the problem sequentially by scanning and processing as required the dataset, in order to aggregate the information. However, this kind of solution incurs considerable communication; besides, it may also be slower. Therefore, this kind of approach is not practical for large datasets, since in this case the central manager becomes a bottleneck.
 
 Our P2PSS algorithm can be briefly described as follows. Each peer processes, by using the Space-Saving algorithm, its local stream of data (or, alternatively, its local dataset) and determines its local frequent items. In order to retrieve the global frequent items, the peers engage in a gossip--based distributed averaging protocol. In each round, they exchange and update their local state, consisting of their Space-Saving stream summary data structure and their current estimate of the number of items in the union of the local streams (or datasets) and of the number of peers in the network.
 
 The contributions of this work are the following ones: (i) we design P2PSS, a fully distributed and gossip--based protocol for frequent items discovery, leveraging the Space-Saving algorithm \cite{Metwally2006}; (ii) we formally prove the correctness and theoretical error bound of P2PSS; (iii) extensive experimental results clearly show that P2PSS provides very good accuracy and scalability.

This paper is organized as follows. We present in Section~\ref{preliminary-defs} preliminary definitions and concepts that shall be used in the rest of the manuscript. Next, we present our P2PSS algorithm in Section~\ref{alg}. We provide an in--depth theoretical analysis of the algorithm, formally proving its correctness and theoretical error bound, in Section~\ref{analysis}. We present and discuss extensive experimental results in Section~\ref{results}, and  recall related work in Section~\ref{related}. Finally, we draw our conclusions in Section~\ref{conclusions}.

\section{Preliminary definitions}
\label{preliminary-defs}

In this Section we introduce preliminary definitions and the notation used throughout the paper. We first introduce the frequent items problem, both in its exact and approximate form, and then we recap the definitions related to the gossip-based protocols.

\subsection{Frequent items problem}
Let $n$ be the number of items in the input $\mathcal{N} = \{s_1,s_2,\ldots,s_n\}$, and $\mathcal{U}=\{1,2,\ldots,m\}$ a universe set from which items are drawn. Therefore, $m = \left| \mathcal{U} \right|$ is the maximum number of possible distinct items in the input. In the sequel, we shall use the notation $[m]$ to denote the set ${1,2,\ldots,m}$.

\begin{defn}
Given an input $\mathcal{N}$ consisting of $n$ elements, the frequency of an item $i\in[m]$ is the number of occurrences of $i$ in $\mathcal{N}$, that is, $f_i = \left\vert\{j\in [n]: s_j = i\}\right\vert$.
\end{defn}

We denote by $\textbf{f} = (f_1,\ldots,f_m)$ the frequency vector, i.e. the vector whose $i$th entry is the frequency of item $i$. It is worth noting here that $||\textbf{f}||_1$, which is the 1-norm of $\textbf{f}$, is by definition the total number of occurrences of all of the items; for this particular setting of the problem, $||\textbf{f}||_1 = n$ (in other settings the input may consist of pairs $\{(s_i, w_i)\}_{i=1,2,\ldots,n}$ where each occurrence $s_i$ is associated to a weight $w_i$; the definition of frequency of an item changes accordingly).

Letting $0 < \phi < 1$ be a support threshold, we can define $\phi$-frequent items as follows.

\begin{defn}
Given an input $\mathcal{N}$ consisting of $n$ elements, and a real value $0<\phi<1$, the $\phi$-frequent items of $\mathcal{N}$ are all those items whose frequency is above $\phi n$, i.e. the elements in the set $F = \{s \in [m]: f_s > \phi n\}$. 
\end{defn}

We are now ready to state the problem of finding the exact $\phi$-frequent items of an input stream.

\begin{prob}
	\label{exactprob}
	(Exact Frequent Items Problem) Given an input $\mathcal{N}$  consisting of $n$ elements and a value $0<\phi<1$, the \textit{Exact Frequent Items Problem} requires finding the set $F = \{s \in [m]: f_s > \phi n\}$ of all the $\phi$-frequent items.
\end{prob}

Problem~\ref{exactprob} is hard or not feasible with limited time and memory resources. In particular, it requires space linear in $n$. Therefore, we shall refer to an approximate version of the problem that accepts the presence of false positives, but can be solved with limited space.

\begin{prob}
\label{approxprob}
(Approximate Frequent Items Problem) Given an input $\mathcal{N}$ consisting of $n$ elements drawn from the universe $[m]$, a value ${0 < \phi < 1}$ and a value ${0<\epsilon<\phi}$, the \textit{Approximate Frequent Items Problem} consists in finding a set $H$, such that:
	\begin{enumerate}
		\item $H$ contains all of the items $s$ with frequency $f_s > \phi n$ ($\phi$-frequent items);		
		\item $H$ does not contain any item $s$ such that ${f_s \leq (\phi-\epsilon)n}$.
	\end{enumerate}
\end{prob}

In this paper, we are concerned with the Approximate Frequent Items Problem in the context of unstructured P2P networks, formally defined as follows.

 \begin{prob}
\label{approxprobp2p}
(Approximate Frequent Items Problem in Unstructured P2P Networks) 
Given an unstructured P2P network consisting of $p$ peers, each peer $l$ must process an input $\mathcal{N}_l$ consisting of $n_l$ elements drawn from the universe $[m]$. Let $n = \sum\limits_{l = 1}^p {{n_l}}$, ${0 < \phi < 1}$ and ${0<\epsilon<\phi}$. The \textit{Approximate Frequent Items Problem in Unstructured P2P Networks} consists in finding a set $H$, such that:
	\begin{enumerate}
		\item $H$ contains all of the items $s$ with frequency $f_s > \phi n$ ($\phi$-frequent items);		
		\item $H$ does not contain any item $s$ such that ${f_s \leq (\phi-\epsilon)n}$.
	\end{enumerate}
\end{prob}

\subsection{Gossip--based protocol}

A gossip--based protocol \cite{Demers:1987} is a synchronous distributed algorithm consisting of periodic rounds. In each of the rounds, a peer (or agent) randomly selects one or more of its neighbours, exchanges its local state with them and finally updates its local state. The information is disseminated through the network by using one of the following possible communication styles: (i) \textit{push}, (ii) \textit{pull} or (iii) \textit{push--pull}. The main difference between push and pull is that in the former a peer randomly selects the peers to whom it wants to send its local state, whilst in the latter it randomly selects the peers from whom to receive the local state. Finally, in the hybrid push--pull communication style, a peer randomly selects the peers to send to and from whom to receive the local state. In this synchronous distributed model it is assumed that updating the local state of a peer is done in constant time, i.e., with $O(1)$ worst-case time complexity; moreover, the duration of a round is such that each peer can complete a push--pull communication within the round. 

We are interested in a specific gossip--based protocol, which is called \textit{distributed averaging}, and can be considered as a consensus protocol. We are given a network of peers described by an undirected graph $G = (V, E)$, where $V = \{1,\ldots,p\}$ is the set of peers' identifiers, and $E$ is the set of edges modelling the communication links between pairs of peers. We assume, for the purpose of our theoretical analysis, that peers and communication links do not fail, and that neither new peers can join the network nor existing peers can leave it (the so-called \textit{churning} phenomenon). Therefore, the graph $G$ describing the underlying network topology is not time-varying. However, it is worth noting here that our algorithm also works in time-varying graphs in which the network can change owing to failures and churning and we shall show an experimental evidence of that in Section~\ref{effect_of_churn}, in which we discuss the effect of churn.

In \textit{uniform gossiping}, a peer $i$ can communicate with a randomly selected peer $j$. Instead, in our scenario the communication among the peers is restricted to neighbour peers i.e., two peers $i$ and $j$ are allowed to communicate if and only if the edge $(i, j) \in E$; we assume that communication links are bidirectional: the existence of the edge $(i, j)$ implies the existence of the edge $(j, i)$. Initially, each peer $i$ is provided with  or computes a real number $v_i$; the distributed averaging problem requires designing a distributed algorithm allowing each peer computing the average $v_{{\rm avg}} = \frac{1}{p}\sum_{i=1}^p v_i$ by exchanging information only with  its neighbours. Letting $v_i(r)$ be the peer $i$ estimated value of $v_{{\rm avg}}$ at round $r$, a gossip interaction between peers $i$ and $j$ updates both peers' variables so that at round $r+1$ it holds that $v_{i}(r+1)=v_j(r+1) = \frac{1}{2}(v_i(r)+v_j(r))$. Of course, for a peer $i$ which is not gossiping at round $r$ it holds that $v_i(r+1) = v_i(r)$. It can be shown that distributed averaging converges exponentially fast to the target value $v_{{\rm avg}}$. In general, a peer is allowed to gossip with at most one peer at a time. In our algorithm, we allow each peer the possibility of gossiping with a predefined number of neighbours. We call \textit{fan-out} $fo$ of peer $i$ the number of its neighbours with which it communicates in each round; therefore, $1 \leq fo \leq \left\vert \{j: (i, j) \in E\} \right\vert $. Therefore, we explicitly allow two  or more  pairs of peers gossiping at the same time, with the constraint that the pairs have no peer in common. We formalize this notion in the following definition.

\begin{defn}
 Two gossip pairs of peers $(i,j)$ and $(x,y)$
  are \textit{noninteracting} if neither $i$ nor $j$ equals either $x$ or $y$.
\end{defn}

In our algorithm multiple non-interacting  pairs of allowable gossips  may occur simultaneously. Non-interactivity is required in order to preserve and guarantee correctness of the results; in the literature non-interactivity is also called \textit{atomic} push--pull communication: given two peers $i$ and $j$, if peer $i$ sends a push message to $j$, then peer $i$ can not receive in the same round any intervening push message from any other peer $k$ before receiving the pull message from $j$ corresponding to its initial push message. 

It is worth noting here that our algorithm do not require explicitly assigning identifiers to the peers, and we do so only for convenience, in order to simplify the analysis; however, we do assume that each peer can distinguish its neighbours.

\section{The P2PSS algorithm}
\label{alg}

The main idea of our P2PSS algorithm is to let each peer determine its local frequent items by processing its local stream of data (or, alternatively, its local dataset) with the Space-Saving algorithm. Then, the peers engage in a gossip--based distributed averaging protocol, exchanging their local state which consists of the Space-Saving stream summary data structure obtained after processing the input stream, and two estimates related respectively to the number of items in the union of the local streams and to the number of peers in the network.

P2PSS is shown as pseudo-code in Algorithm~\ref{p2pssalg}. It consists of several procedures. The \textsc{initialization} procedure requires the following parameters: $l$, the peer's identifier; $\mathcal{N}_l$, the local dataset to be processed by peer $l$; $C$, the convergence factor (whose role shall be explained in Section~\ref{jelasity}); $k$, the number of counters to be used for the Space-Saving stream summary data structure; $R$, the number of rounds to be performed by the distributed algorithm; $p^*$, an estimate of the number of peers in the network (we only require $p^* \geq p$); $\phi$, the threshold to be used to determine the frequent items; $\epsilon$, the error tolerance and $0 < \delta < 1$, the probability of failure of the algorithm. Each peer $l$ initializes a Space-Saving stream summary data structure with $k$ counters, sets the current round $r$ to zero and its estimate $\tilde{n}_{r,l}$ of the average number of items over all of the peers to the number of items in its local dataset. The variable $\tilde{n}_{r,l}$ is therefore an estimate for the quantity $\bar{n} = \frac{1}{p} \sum_{l=1}^p \left\vert\mathcal{N}_l\right\vert$. Then, the peer whose identifier is 1 sets $\tilde{q}_{r,l}$ to 1 and all of the other peers sets this value to zero. The variable $\tilde{q}_{r,l}$ is used to estimate the number $p$ of peers by using the distributed averaging protocol: indeed, upon convergence this value approaches with high probability $1/p$. Next, each peer processes its local dataset  $\mathcal{N}_l$ by using the Space-Saving algorithm, obtaining as a result the stream summary $\mathcal{S}_{r,l}$ containing its local frequent items. It is worth noting here that $\mathcal{N}_l$ does not need to be a locally stored dataset: indeed, the input can be a stream and, as such, its items may be processed one at a time in a streaming fashion, without requiring explicitly local storage. The peer local state is a tuple $state_{r,l}$ consisting of the peer's local summary $\mathcal{S}_{r,l}$, and the estimates $\tilde{n}_{r,l}$ and $\tilde{q}_{r,l}$.

The \textsc{gossip} procedure lasts for $R$ rounds. During each round a peer increments $r$, the current round, selects $fo$ (the fan-out) neighbours uniformly at random and sends to each of them its local state in a message of type \textit{push}. Upon receiving a message, each peer executes the \textsc{on\_receive} procedure. From the message, the peer extracts the message's type, sender and state sent. A message is processed accordingly to its type as follows. A push message is handled in two steps. In the first one, the peer updates its local state by using the state received; this is done by invoking the \textsc{update} procedure that we shall describe later. In the second one, the peer sends back to the sender, in a message of type \textit{pull}, its updated local state. A pull message is handled by a peer setting its local state equal to the state received.

\begin{algorithm}
\begin{algorithmic}[1]
\caption{\textsc{P2PSS}: P2P Space-Saving}	
\label{p2pssalg}

\Procedure{Initialization}{$l$, $\mathcal{N}_l$, $C$, $k$, $R$, $f$, $p^*$, $\phi$, $\epsilon$, $\delta$}
\Comment{initialization of node $l$}
\State $r \leftarrow 0$
\State $\tilde{n}_{r,l} \leftarrow \left\vert\mathcal{N}_l\right\vert$
\If{$l == 1$}
	\State $\tilde{q}_{r,l} \leftarrow 1$
\Else
	\State $\tilde{q}_{r,l} \leftarrow 0$
\EndIf
\State $\mathcal{S}_{r,l} \leftarrow$ \Call{SpaceSaving}{$\mathcal{N}_l$, $k$}
\State $state_{r,l} \leftarrow (\mathcal{S}_{r,l}, \tilde{n}_{r,l}, \tilde{q}_{r,l})$
\EndProcedure

\Procedure{GOSSIP}{}
\For{$r = 0$ to $R$}
	\State $neighbours \leftarrow$ select $fo$ random neighbours
	\ForAll {$i \in neighbours$}
		\State \Call{SEND}{$push$, $i$, $state_{r,l}$}
	\EndFor
\EndFor
\EndProcedure

\Procedure{ON\_RECEIVE}{$msg$}
\State $type \leftarrow msg.type$
\State $j \leftarrow msg.sender$
\State $state \leftarrow msg.state$
\If{$type == push$}
	\State $state_{r+1,l} \leftarrow$ \Call{UPDATE}{$state$, $state_{r,l}$}
	\State \Call{SEND}{$pull$, $j$, $state_{r+1,l}$}
\EndIf
\If{$type == pull$}
	\State $state_{r+1,l} \leftarrow state$
\EndIf
\EndProcedure

\Procedure{QUERY}{}
\State $(\mathcal{S}_{r,l}, \tilde{n}_{r,l}, \tilde{q}_{r,l}) \leftarrow state_{r,l}$
\State $\epsilon^* \leftarrow p^* \times \sqrt{\frac{C^{r}}{\delta}}$
\State $t \leftarrow \phi \tilde{n}_{r,l} \frac{1-\epsilon^*}{1+\epsilon^*}$
\State $\tilde{p}_{r,l} \leftarrow 1/\tilde{q}_{r,l}$
\State $H \leftarrow \emptyset$
\ForAll{counter $c \in \mathcal{S}_{r,l}$}
	\If{$c.f > t$}
		\State $H \leftarrow H \cup (c.i, c.f \times \tilde{p}_{r,l})$
	\EndIf
\EndFor
\State \Return $H$
\EndProcedure
\end{algorithmic}	
\end{algorithm}

The \textsc{update} procedure, shown in pseudo-code as Algorithm~\ref{update}, works as follows: the two local summaries of peers $i$ and $j$ are merged by invoking the \textsc{merge} procedure reported in Algorithm~\ref{merge}, producing the stream summary $\mathcal{S}$; since we want to implement a distributed averaging protocol, we scan the counters of the stream summary $\mathcal{S}$, and for each counter $c$ we update its frequency $c.f$ dividing it by 2; finally, we compute as required by the averaging protocol the estimates $\tilde{n}$ and $\tilde{q}$ and return the updated state just computed.

Here we briefly recap how merging works: for each item belonging to both the local summaries of peers $i$ and $j$, we insert the item in the output stream summary with an estimated frequency equal to the sum of its estimated frequencies in the two input summaries. If an item belongs to just one of the summaries, its estimated frequency in the output stream summary is equal instead to the sum of its estimated frequency and the minimum estimated frequency in the other summary. Finally, if the output stream summary contains more than $k$ counters (the output summary may contain at most $2k$ items; this happens when all of the items in both summaries are distinct), we prune the summary and return as output summary only the first $k$ items with the greatest estimated frequencies, otherwise we return the output summary as is.

Finally, the user can issue a \textsc{query} procedure to an arbitrary peer to retrieve the frequent items determined by our algorithm. This is done by computing $t$, a threshold that determines whether an item is a candidate frequent or not, and $\tilde{p}_{r,l}$, the estimate of $p$. Note that $t$ is defined in terms of $\epsilon^*$, whose meaning shall be explained in the Section devoted to the theoretical analysis of the algorithm. Then, we initialize $H$ to an empty set and scan each of the counters in the local stream summary $\mathcal{S}_{r,l}$, checking whether the frequency $c.f$ of the item $c.i$ stored in the counter $c$ is greater than the threshold $t$ or not. For each item which is determined to be candidate frequent, we add the tuple $(c.i, c.f \times \tilde{p}_{r,l})$ to $H$ and finally we return $H$.

\begin{algorithm}
\begin{algorithmic}[1]
\caption{\textsc{UPDATE}: Update procedure}
\label{update}

\Procedure{UPDATE}{$state_i$, $state_j$}
\State $(\mathcal{S}_i, \tilde{n}_i, \tilde{q}_i) \leftarrow state_i$
\State $(\mathcal{S}_j, \tilde{n}_j, \tilde{q}_j) \leftarrow state_j$
\State $\mathcal{S} \leftarrow$ \Call{MERGE}{$\mathcal{S}_i$, $\mathcal{S}_j$}
\ForAll{counter $c \in \mathcal{S}$}
	\State $c.f \leftarrow \frac{c.f}{2}$
\EndFor	
	
\State $\tilde{n} \leftarrow \frac{\tilde{n}_i + \tilde{n}_j}{2}$  
\State $\tilde{q} \leftarrow \frac{\tilde{q}_i + \tilde{q}_j}{2}$  
\State $state \leftarrow (\mathcal{S}, \tilde{n}, \tilde{q})$
\State \Return $state$
\EndProcedure

\end{algorithmic}
\end{algorithm}

\begin{algorithm}
	\caption{Merge}
	\label{merge}
	\begin{algorithmic}
		
		\Require {$\mathcal{S}_1$, $\mathcal{S}_2$: vector representing summaries of $k$ counters ordered by item's frequency; $k$, number of counters in each summary;}
		
		\State $m_1 \leftarrow \mathcal{S}_1[0].\hat f $
		\Comment{minimum of all of the frequencies in $\mathcal{S}_1$}
		\State $m_2 \leftarrow \mathcal{S}_2[0].\hat f $
		\Comment {minimum of all of the frequencies in $\mathcal{S}_2$}
		
		\State $\mathcal{S}_M \leftarrow \emptyset$
		\ForAll{counter $\mathcal{S}_1[j]$ in $\mathcal{S}_1$}
		\State $new\_counter.i \leftarrow \mathcal{S}_1[j].i$	
		\State $counter_{\mathcal{S}_2} \leftarrow \mathcal{S}_2.$\Call{Find}{$\mathcal{S}_1[j].i$}
		\If{$counter_{\mathcal{S}_2}$}
		\State $new\_counter.\hat{f} \leftarrow \frac{1}{2} \left( \mathcal{S}_1[j].\hat{f} + counter_{\mathcal{S}_2}.\hat{f} \right)$
		\State $\mathcal{S}_2.$\Call{Remove}{$counter_{\mathcal{S}_2}$}
		\Else
		\State $new\_counter.\hat{f} \leftarrow \frac{1}{2} \left( \mathcal{S}_1[j].\hat{f} + m_2 \right )$
		\EndIf
		\State $\mathcal{S}_M.$\Call{Put}{$new\_counter$}
		\EndFor
		
		\ForAll{counter $\mathcal{S}_2[j]$ in $\mathcal{S}_2$}
		\State $new\_counter.i \leftarrow \mathcal{S}_2[j].i$
		\State $new\_counter.\hat{f} \leftarrow \frac{1}{2} \left( \mathcal{S}_2[j].\hat{f} + m_1 \right )$
		\State $\mathcal{S}_M$.\Call{Put}{$new\_counter$}
		\EndFor
		
		\State $\mathcal{S}_M$.\Call{Prune}{$k$}
		\Comment{Select $k$ counters with the greatest frequencies and delete the others}
		\State \Return $\mathcal{S}_M$
		
	\end{algorithmic}
\end{algorithm}

To better explain the P2PSS algorithm, we propose and discuss an example. Let us suppose that there are 4 peers, each with a stream summary holding 4 counters. Figure~\ref{example_a} shows the state of the stream summary for each peer before starting the gossip protocol. For each item (identified by a letter), its frequency is reported. Suppose that, during the first round, peer $p_0$ exchanges data with $p_1$ and peer $p_2$ with $p_3$. Figure~\ref{example_b} depicts the peers' stream summaries at the end of the first round. Supposing that in the second round $p_1$ exchanges data with $p_2$ and $p_0$ with $p_3$, Figure~\ref{example_c} provides the state of the stream summaries converged to the final values. Each summary reports the average estimate (with regard to the number of peers) of the items' frequency.

\begin{figure*}[h]
	\centering
	\begin{tabular}{ccc}		
		\subfloat[Initial state]{
			\includegraphics[width=0.3\textwidth]{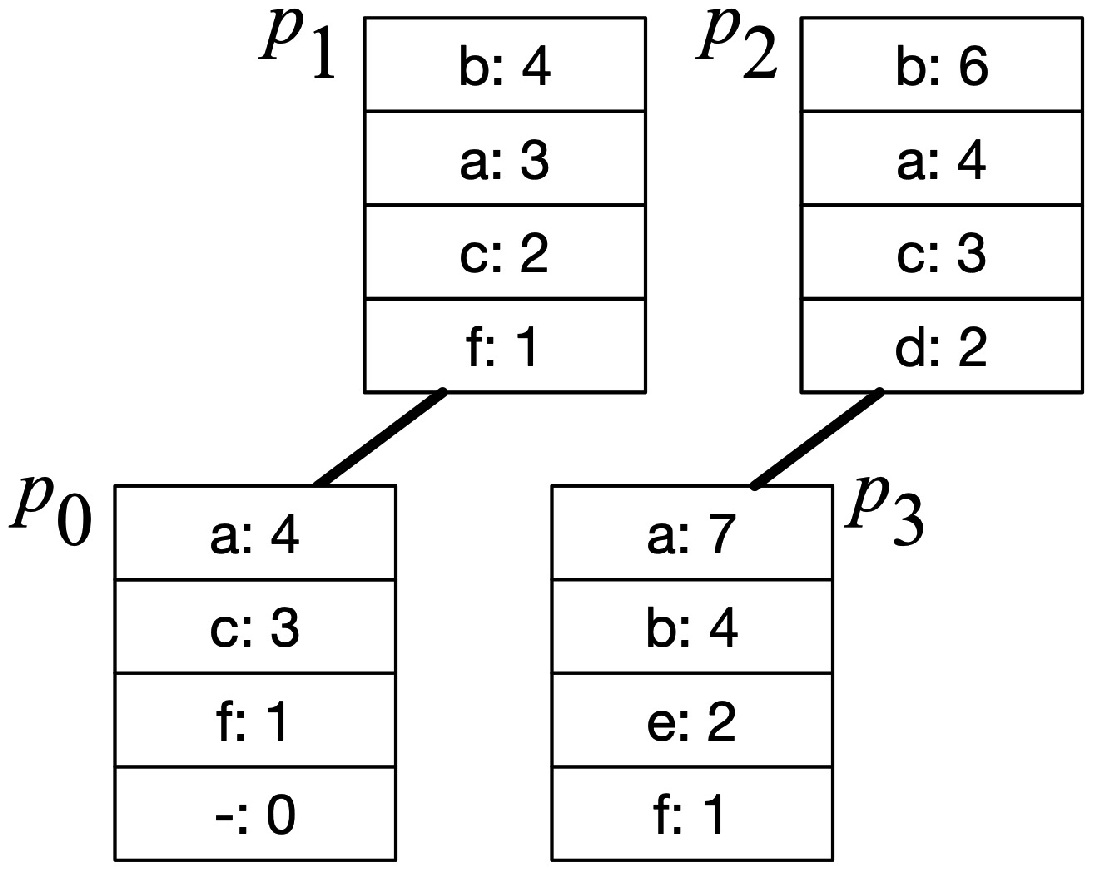}
			\label{example_a}
		} &
		
		\subfloat[State after first round]{
			\includegraphics[width=0.3\textwidth]{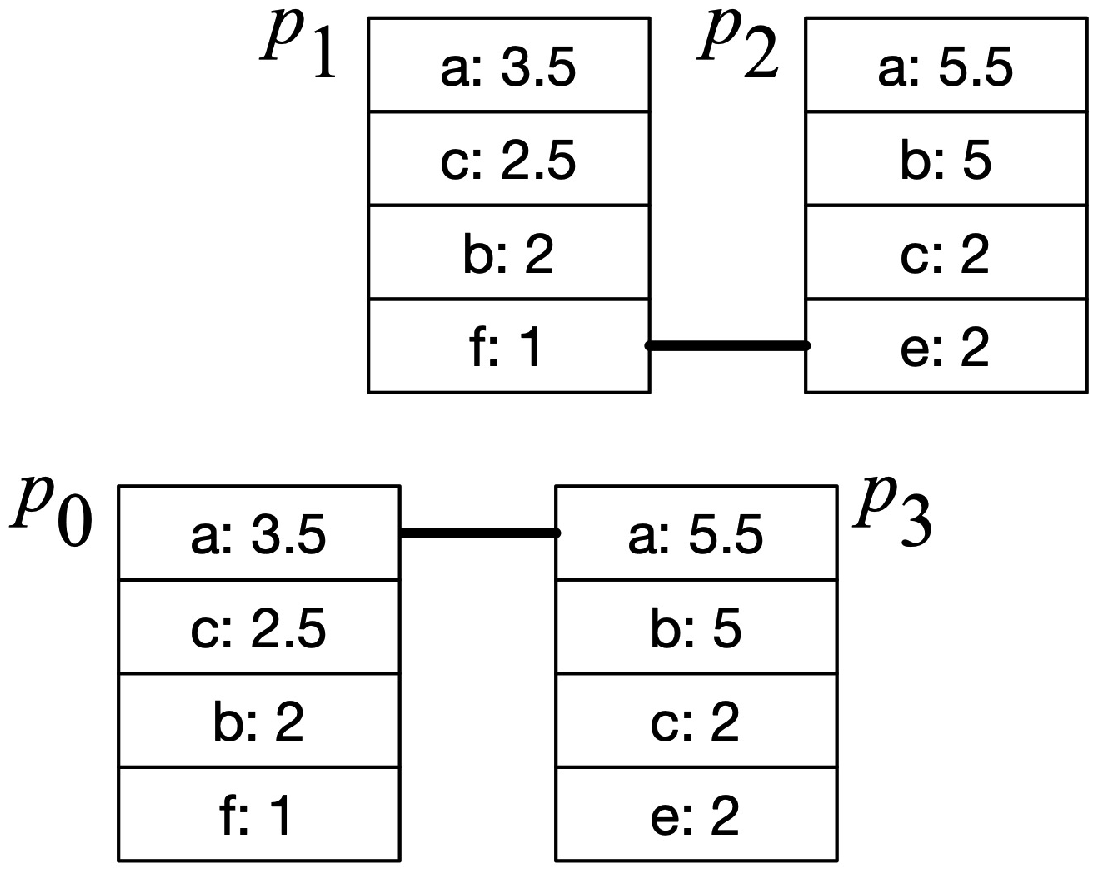}
			\label{example_b}
		} &
	  	
			\subfloat[Final state ]{
			\includegraphics[width=0.3\textwidth]{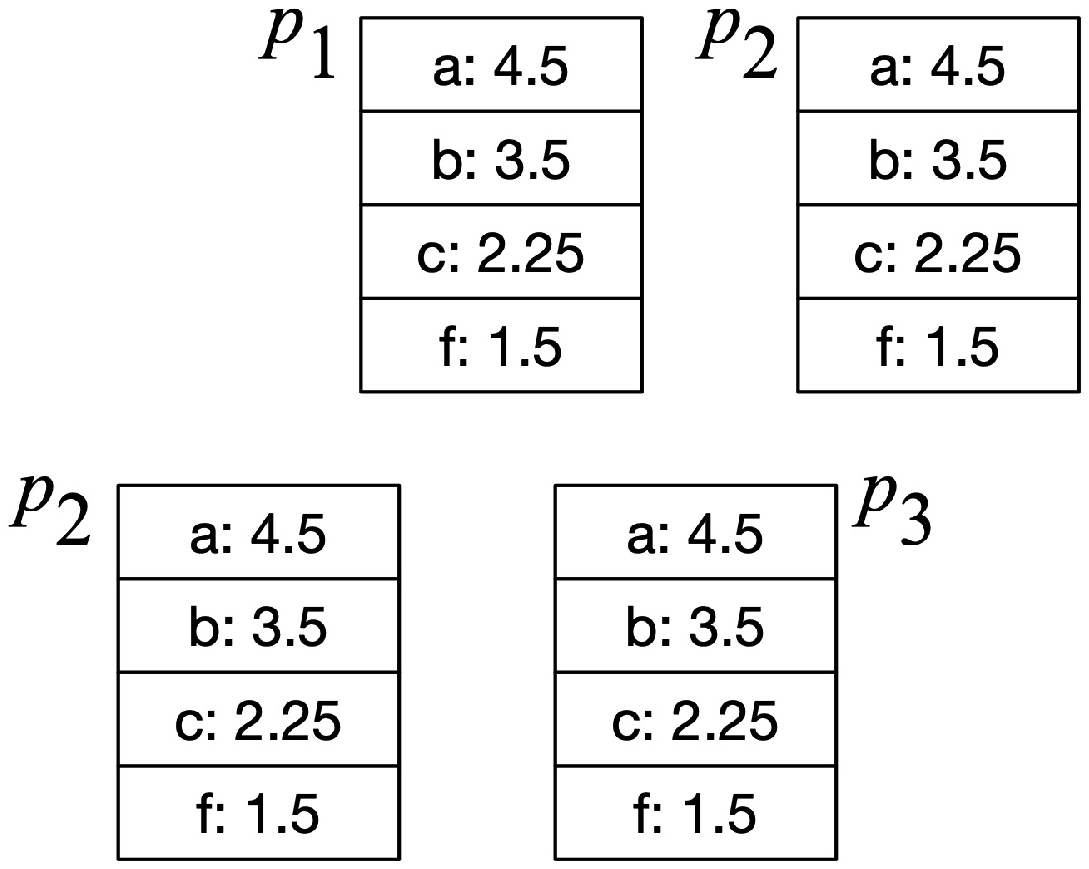}
			\label{example_c}
		} 
	\end{tabular}
	
	\caption{Example of P2PSS algorithm acting over 4 peers and a stream summary with 4 counters.} 
	\label{example}
\end{figure*}

\section {Theoretical analysis}
\label{analysis}

Before proceeding with our analysis, we need to recall the results by Jelasity et al. in \cite{Jelasity2005} on which we rely for our discussion. 
Jelasity et al. in the cited paper propose a gossip--based algorithm for computing the average value of numbers held by the nodes of a network. They show that the algorithm converges to the true average value and give an estimation of its convergence factor. Their reasoning is based on a centralized algorithm operating globally on the distributed state of the system that allows simplifying the theoretical analysis by conveniently simulating the gossip--based distributed version of the algorithm. Even though the analysis of \cite{Jelasity2005} relies on uniform gossiping (i.e., the underlying topology is described by a complete graph), there is no significant difference between the performance of randomized gossiping in complete graphs and sparse random graphs \cite{7161540} \cite{5462084} (this has been experimentally verified by Jelasity et al.). Therefore, in this Section we shall follow the Jelasity et al. strategy and show that \textsc{P2PSS} also converges and correctly solves the Approximate Frequent Items Problem in Unstructured P2P Networks.

\subsection{Jelasity's averaging algorithm}
\label{jelasity}

The centralized \textsc{AVG} algorithm by Jelasity et al., takes a vector $\boldsymbol{w}_r$ of length $p$ representing the state of the nodes after the $r$th round ($p$ is the number of nodes in the network and each component of the vector is a value held by a node) and produces a new vector $\boldsymbol{w}_{r+1} = \textsc{AVG}(\boldsymbol{w}_r)$ of the same length, representing the state of the system after another round of gossip. At each elementary step of \textsc{AVG}, two selected nodes update their state so that the vector $\boldsymbol{w}_r$ becomes:
\begin{equation}
\begin{split}
\boldsymbol{w}_{r}' = (w_{r,1}, w_{r,2}, \ldots, \frac{w_{r,i}+w_{r,j}}{2}, \ldots, \\
 \frac{w_{r,i}+w_{r,j}}{2}, \ldots, w_{r,p} ).
\end{split}
\end{equation}
\noindent After $p$ elementary steps \textsc{AVG} returns the vector $\boldsymbol{w}_{r+1}$. Through a proper selection of the pair of nodes, this algorithm can reproduce the behavior of the distributed gossip--based averaging algorithm introduced by Jelasity et al., since each call to \textsc{AVG} corresponds to a round of that algorithm. We refer the interested reader to \cite{Jelasity2005} for all of the details. 

Here, we only recall the results essential for our purposes. The averaging protocol proposed by Jelasity et al. and its centralized equivalent can be seen as variance reduction algorithms. Consider a variance measure $\sigma_r^2$ defined as:
\begin{equation}
\label{eq_sigma2}
\sigma_r^2 = \frac{1}{p-1} \sum_{l=1}^{p}\left(w_{r,l} - \bar{w}\right)^2,
\end{equation}
where $w_{r,l}$ is the value held by peer $l$ after $r$ rounds of the gossip algorithm and $\bar{w} = \frac{1}{p}\sum_{l=1}^p w_{0,l}$ is the mean of the initial values held by the peers. The authors in \cite{Jelasity2005} state that, if $\psi_k$ is a random variable denoting the number of times a node $k$ is chosen as a member of the pair of nodes exchanging their states during a round of the protocol, and each pair of values $w_{r,i}$ and $w_{r,j}$ selected by each call to \textsc{GetPair} are uncorrelated, then the following theorem holds.
\begin{thm}
	\label{Jelasity}
	{\rm \cite{Jelasity2005}} If \textsc{GetPair} has the following properties:
	\begin{enumerate}
		\item the random variables $\psi_1, \ldots , \psi_p$ are identically distributed (let $\psi$ denotes a random variable with this common distribution),
		\item after $(i, j)$ is returned by \textsc{GetPair}, the number of times $i$ and $j$ shall be selected by the remaining calls to \textsc{GetPair} have identical distributions,
	\end{enumerate}
	then we have:
	\begin{equation}
	\label{eq_conv_factor}
	\mathbb{E}[\sigma_{r+1}^2] \approx \mathbb{E}[2^{-\psi}] \mathbb{E}[\sigma_{r}^2].
	\end{equation}
\end{thm}
\vspace{2mm}
The random variable $\psi$ only depends on the particular implementation of \textsc{GetPair}. From eq. \eqref{eq_conv_factor}, the convergence factor is defined as:
\begin{equation}
\frac{\mathbb{E}[\sigma_{r+1}^2]}{ \mathbb{E}[\sigma_{r}^2]} = \mathbb{E}[2^{-\psi}];
\end{equation}

Therefore, the convergence factor depends on $\psi$ and, as a consequence, on the pair selection method. Jelasity et al. compute the convergence factor for different implementations of the pair selection method, but we are only interested in the one which allows simulating the distributed gossip--based averaging protocol, which they call \textsc{GetPair\_Distr}. This method consists in drawing a random permutation of the nodes and then, for each node in that permutation, choosing another random node in order to form a pair. For this selection method, the convergence factor is $C = \mathbb{E}[2^{-\psi}] = 1/(2\sqrt{e})$.

We now derive from Theorem~\ref{Jelasity} the following proposition.

\begin{prop}
	\label{prop1}
	Let $\delta$ be a user-defined probability, $w_{r,l}$ the value held by peer $l$ after $r$ rounds of the averaging protocol, $p$ the number of peers participating in the protocol, $C = \mathbb{E}[2^{-\psi}] = 1/(2\sqrt{e})$ the convergence factor and $\bar{w}$ the mean of the initial vector of values $\boldsymbol{w}_0$, i.e. $\bar{w} = 1/p \sum_{l=1}^{p} w_{0,l}$. Then, with probability $1 - \delta$ it holds that, for any peer $l$:
	\begin{equation}
	\label{eq_error_due_to_gossip}
	\left|w_{r,l} - \bar{w}\right| < \sqrt{(p-1) \sigma_0^2} \sqrt{\frac{C^r }{\delta}}
	\end{equation}
\end{prop}

\begin{proof}
	
	From eq. \eqref{eq_conv_factor} it follows that:
	
	\begin{equation}
	\label{eq_expected_sigma2}
	\mathbb{E}[\sigma_{r}^2] = \mathbb{E}[2^{-\psi}]^r \sigma_0^2;
	\end{equation}
	
	where $\sigma_0^2$ depends on the distribution of the initial numbers among the peers. Through the Markov inequality, we have that:
	
	\begin{equation}
	\mathbb{P}[\sigma_{r}^2 \ge \frac{\mathbb{E}[\sigma_{r}^2]}{\delta}] \le  \delta;
	\end{equation}
	
	\noindent or
	
	\begin{equation}
	\mathbb{P}[\sigma_{r}^2 < \frac{\mathbb{E}[\sigma_{r}^2]}{\delta}] \ge  1 - \delta.
	\end{equation}
	
	Considering eqs. \eqref{eq_sigma2} and \eqref{eq_expected_sigma2}, it holds that:
	
	\begin{equation}
	\mathbb{P}[\sum_{l=1}^{p}\left(w_{r,l} - \bar{w}\right)^2 < (p-1) \frac{C^r \sigma_0^2}{\delta}] \ge  1 - \delta.
	\end{equation}
	
	As a consequence, with probability at least $1-\delta$:
	
	\begin{equation}
	max_{l \in [p]} \left(w_{r,l} - \bar{w}\right)^2 \le \sum_{l=1}^{p}\left(w_{r,l} - \bar{w}\right)^2 < (p-1) \frac{C^r \sigma_0^2}{\delta},
	\end{equation}
	
	\noindent which implies:
	
	\begin{equation}
	max_{l \in [p]} \left|w_{r,l} - \bar{w}\right| < \sqrt{(p-1)\sigma_0^2} \sqrt{\frac{C^r }{\delta}}.
	\end{equation}
	
	This proves the proposition.	
\end{proof}

Eq. \eqref{eq_error_due_to_gossip} gives an upper bound on the error made by any peer in estimating the value $\bar{w}$ after $r$ rounds of the Jelasity's averaging algorithm. This bound is probabilistic and it is valid with probability greater than or equal to $1-\delta$.

\subsection{Merging of Space-Saving summaries}
\label{merge_section}
\textsc{P2PSS} follows the same structure of the gossip--based averaging protocol by Jelasity et al., but it is based on the procedure introduced by Cafaro et al. in \cite{cafaro-pulimeno-tempesta} in order to merge Space-Saving summaries.  The Merge algorithm has been introduced in Section~\ref{alg}, here we briefly recap its properties. We shall use multisets to represent both the input streams and the corresponding summaries. 

\begin{defn}
	\label{multiset}
	A multiset $\mathcal{N}=(N, f_{\mathcal{N}})$ is a pair where $N$ is some set, called the underlying set of elements, and $f_{\mathcal{N}}: N \rightarrow \mathbb{N}$ is a function.
	The generalized indicator function of $\mathcal{N}$ is 
	\begin{equation}
	I_\mathcal{N} (x) := \left\{ {\begin{array}{*{20}c}
		{f_{\mathcal{N}}(x)} & {x \in N} , \\
		0 & {x \notin N},  \\
		\end{array} }\right.
	\end{equation}
	
	\noindent where the integer--valued function $f_{\mathcal{N}}$, for each $x \in N$, provides its  \textit{frequency} (or multiplicity), i.e., the number of occurrences of $x$ in $\mathcal{N}$. 
	\noi The cardinality of $\mathcal{N}$ is expressed by

	\begin{equation}
	\left\vert{\mathcal{N}}\right\vert := Card(\mathcal{N}) = \sum\limits_{x \in N} {I_\mathcal{N} (x)},
	\end{equation}
	
	\noi whilst the cardinality of the underlying set $N$ is
	
	\begin{equation}
	\left\vert{N}\right\vert := Card(N) = \sum\limits_{x \in N} {1}.
	\end{equation}
\end{defn}

A multiset, or bag,  is defined by a proper set (the support set) and a multiplicity function: it is a set where elements can be repeated, i.e., an element in a multiset can have multiplicity greater than one. 

Let $\mathcal{U} = [d]$ be the universe from which the items in input are drawn and let $\mathcal{N}_1 = (N_1, f_{\mathcal{N}_1})$ and $\mathcal{N}_2~=~(N_2,~f_{\mathcal{N}_2})$ be two input multisets, where $N_i\subseteq \mathcal{U}$ for $i=1,2$. Furthermore, let $\mathcal{S}_1 = (\Sigma_1, \hat{f}_{\mathcal{S}_1})$ and $\mathcal{S}_2 = (\Sigma_2, \hat{f}_{\mathcal{S}_2})$ be two Space-Saving summaries with at most $k$ distinct items, corresponding respectively to $\mathcal{N}_1$ and $\mathcal{N}_2$.
Let $\oplus_k$ be the merge operation described in \cite{cafaro-pulimeno-tempesta} and shown in pseudo-code as Algorithm~\ref{merge}, where subscript $k$ indicates the maximum number of distinct items in each involved summary. Then, the summary $\mathcal{S}_M = \mathcal{S}_1 \oplus_k \mathcal{S}_2$ is a summary for  $\mathcal{N} = \mathcal{N}_1 \uplus \mathcal{N}_2$ with at most $k$ distinct items that continues to guarantee the same bounds on size and error of the original summaries. In particular, the following relations hold, for each item $e \in N$, being $\hat{f}_{\mathcal{S}_M}^{min}$ the minimum frequency in $\mathcal{S}_M$ and  $\hat{f}_{\mathcal{S}_M}^{min} = 0$ when  $\left\vert{\Sigma}_M\right\vert < k$.

\begin{equation}
\label{ssprop1}
\left\vert{\mathcal{S}_M}\right\vert \leq \left\vert{\mathcal{N}}\right\vert,
\end{equation}

\begin{equation}
\label{ssprop2}
\hat{f}_{\mathcal{S}_M}(e) - \hat{f}_{\mathcal{S}_M}^{min} \leq f_\mathcal{N}(e) \leq \hat{f}_{\mathcal{S}
	_M}(e),  \qquad e \in \Sigma_M,
\end{equation}

\begin{equation}
\label{ssprop3}
f_\mathcal{N}(e)  \leq \hat{f}_{\mathcal{S}_M}^{min}, \qquad e \notin \Sigma_M,
\end{equation}

\begin{equation}
\label{ssprop4}
\hat{f}_{\mathcal{S}_M}^{min}  \leq \left\lfloor\frac{\left\vert{\mathcal{N}}\right\vert}{k}\right\rfloor.
\end{equation}

\vspace{3mm}
The properties in eqs. \eqref{ssprop1}--\eqref{ssprop4} guarantee that if  $\mathcal{S}_1$ and $\mathcal{S}_2$ respect the same properties (and it has been proven that Space-Saving summaries do), then $\mathcal{S}_M$ contains all of the $\phi$-frequent items of $\mathcal{N}$ with $\phi > 1/k$ and solves the Approximate Frequent Items Problem in Unstructured P2P Networks with tolerance $\epsilon = 1/k$. 

\subsection{Convergence of P2PSS}

Let $\mathcal{M}$ be the class of all the multisets with support set included in $ \mathcal{U}$. 
Let us introduce the operation $\oslash_d: \mathcal{M} \rightarrow \mathcal{M}$, so that $\oslash_d(\mathcal{N}) = (N, f_N/d))$, i.e, the multiset $\oslash_d(\mathcal{N})$ has the same support set of $\mathcal{N}$, but each  element has a fraction $1/d$ of the multiplicity it has in $\mathcal{N}$, where we explicitly allow for fractional multiplicities. We have that $\biguplus_{i=1}^d\oslash_d(\mathcal{N}) =  \oslash_{\frac{1}{d}}(\oslash_d(\mathcal{N})) = \mathcal{N}$ and it is immediate to see that if $\mathcal{S}$ is a summary for $\mathcal{N}$, then $\oslash_d(\mathcal{S})$ is a summary for $\oslash_d(\mathcal{N})$. In fact, if we divide by $d$ all of the terms in eqs. \eqref{ssprop1}--\eqref{ssprop4}, the same relations continue to hold. Furthermore, it holds that $\biguplus_{i=1}^d\oslash_d(\mathcal{S}) = \oslash_{\frac{1}{d}}(\oslash_d(\mathcal{S})) = \mathcal{S}$.

Following the Jelasity et al. approach, we introduce \textsc{AVG-Merge} as Algorithm~\ref{algo1}. This is a centralized algorithm that simulates the distributed \textsc{P2PSS} algorithm. \textsc{AVG-Merge}, through the selection method \textsc{GetPair\_Distr}, operates on the global state of the network by simulating the distributed \textsc{P2PSS} protocol and allowing us to simplify its theoretical analysis.

\begin{algorithm}
	\begin{algorithmic}
		\caption{\textsc{AVG-Merge}: global Space-Saving summaries average merging}
		\label{algo1}
		\Require{$\boldsymbol{\mathcal{S}}_r = (\mathcal{S}_{r,1}, \mathcal{S}_{r,2}, \ldots, \mathcal{S}_{r,p})$: a vector of Space-Saving summaries, $k$: the maximum number of distinct items in each summary, $p$: the number of peers}
		\State $l \leftarrow 0$
		\While{$l < p$}
		\State $(i, j) \leftarrow $ \Call{GetPair}{\ } 	
		\State $\mathcal{S}_{r,i} \leftarrow \mathcal{S}_{r,j} \leftarrow  \oslash_2(\mathcal{S}_{r,i} \oplus_k \mathcal{S}_{r,j})$	
		\State $l \leftarrow l + 1$
		\EndWhile
		\State \Return $\boldsymbol{\mathcal{S}}_r$ as $\boldsymbol{\mathcal{S}}_{r+1}$ 
		
	\end{algorithmic}
\end{algorithm}

\vspace{4mm}
Algorithm~\ref{algo1} is similar to \textsc{AVG} algorithm discussed in Section~\ref{jelasity}, but it operates on multisets rather than single values. Initially, each peer computes a local summary on its input stream, through the execution of Space-Saving with $k$ counters, then the distributed protocol starts. 

The initial distributed state of the system can be represented by the vector of the local summaries $\boldsymbol{\mathcal{S}}_0~=~(\mathcal{S}_{0,1},\mathcal{S}_{0,2},\ldots,\mathcal{S}_{0,p})$, where $p$ is the number of peers participating in the protocol. Another vector is naturally associated to $\boldsymbol{\mathcal{S}}_0$: the vector of the local input streams $\boldsymbol{\mathcal{N}}_0 = (\mathcal{N}_{0,1}, \mathcal{N}_{0,2}, \ldots, \mathcal{N}_{0,p})$. We have that $\biguplus_{l = 1}^p \mathcal{N}_{0,l} = \mathcal{N}$, where we denote by $\mathcal{N}$ the global input stream.

Each call to \textsc{AVG-Merge} corresponds to a round of \textsc{P2PSS}. It modifies $\boldsymbol{\mathcal{S}}_r$, the vector of the summaries held by the peers at the end of round $r$, producing the vector $\boldsymbol{\mathcal{S}}_{r+1}$. Furthermore,  implicitly also $\boldsymbol{\mathcal{N}}_{r}$, the vector of local input streams to which the summaries refer, changes to $\boldsymbol{\mathcal{N}}_{r+1}$. In fact, let $\boldsymbol{\mathcal{S}}_r$ and $\boldsymbol{\mathcal{N}}_r$ be the vectors of the summaries owned by each peer and the corresponding partitions of the input stream $\mathcal{N}$ after the $r$th round. Then, after each iteration of the main loop of \textsc{AVG-Merge}, letting $(i, j)$ be the pair of communicating peers, i.e. the pair selected by \textsc{GetPair}, the vector of summaries becomes: 

\begin{equation}
\begin{split}
\boldsymbol{\mathcal{S}}_{r}' = (\mathcal{S}_{r,1}, \mathcal{S}_{r,2}, \ldots, \oslash_2(\mathcal{S}_{r,i} \oplus_k \mathcal{S}_{r,j}), \ldots,\\
 \oslash_2(\mathcal{S}_{r,i} \oplus_k \mathcal{S}_{r,j}), \ldots, \mathcal{S}_{r,p}),
\end{split}
\end{equation}

\noindent and the corresponding vector of partitions of the input stream shall change to:

\begin{equation}
\label{elem-it-msets}
\begin{split}
\boldsymbol{\mathcal{N}}_{r}' = (\mathcal{N}_{r,1}, \mathcal{N}_{r,2}, \ldots, \oslash_2(\mathcal{N}_{r,i} \uplus \mathcal{N}_{r,j}), \ldots, \\
\oslash_2(\mathcal{N}_{r,i} \uplus \mathcal{N}_{r,j}), \ldots, \mathcal{N}_{r,p}).
\end{split}
\end{equation}

From what we said on the operations $\oplus$ and $\oslash$, after each elementary iteration of \textsc{AVG-Merge}, two invariants hold: 
\begin{enumerate}
	\item each peer $l$ owns a summary $\mathcal{S}_{r,l}$ which is a correct Space-Saving summary for the portion of input stream $\mathcal{N}_{r,l}$;
	\item  $\biguplus_{l = 1}^p \mathcal{N}_{r,l} = \mathcal{N}$.
\end{enumerate}

These invariants remain true after each iteration of the main loop of \textsc{AVG-Merge} and, consequently, after each call to \textsc{AVG-Merge}, that is after each round of the \textsc{P2PSS} distributed protocol, when we derive from the vectors $\boldsymbol{\mathcal{S}}_r$ and $\boldsymbol{\mathcal{N}}_r$, the new vectors $\boldsymbol{\mathcal{S}}_{r+1}$ and $\boldsymbol{\mathcal{N}}_{r+1}$. 

We can state that, for $r \rightarrow \infty$, the two vectors $\boldsymbol{\mathcal{S}}_r$ and $\boldsymbol{\mathcal{N}}_r$ converge respectively to:

\begin{equation}
\boldsymbol{\mathcal{S}}_{\infty} = \left(\mathcal{S}_{\rm avg}, \mathcal{S}_{\rm avg}, \ldots, \mathcal{S}_{\rm avg}\right)
\end{equation}

\noindent and 

\begin{equation}
\label{conv-msets}
\boldsymbol{\mathcal{N}}_{\infty} = \left(\mathcal{N}_{\rm avg}, \mathcal{N}_{\rm avg}, \ldots, \mathcal{N}_{\rm avg}\right),
\end{equation}

\noindent where $\mathcal{N}_{\rm avg} = \oslash_p(\mathcal{N})$ and $\mathcal{S}_{\rm avg}$ is a correct summary of $\mathcal{N}_{\rm avg}$. 

This means that all of the peers converge to a summary of $\oslash_p(\mathcal{N})$, from which, for the properties of the operations $\oplus$ and $\oslash$, a correct summary for $\mathcal{N}$ can be derived by computing $\oslash_\frac{1}{p}(\mathcal{S}_{\rm avg})$ (we actually need to know the number of peers, which is not always the case, but we shall see in the following how we can estimate $p$), i.e. \textsc{P2PSS} converges. 

Thanks to the invariants discussed above, in order to prove the convergence of the summaries to $\mathcal{S}_{\rm avg}$,  it's enough to verify that the  local input streams implicitly induced by the algorithm converge to $\mathcal{N}_{\rm avg}$.

We can represent each initial  local input stream $\mathcal{N}_{0,l}$ for $l = 1,2, \ldots, p$, as the frequencies' vector of the items in that stream, $\boldsymbol{\tilde{f}}_{0,l} = (\tilde{f}_{0,l,1}, \tilde{f}_{0,l,2}, \ldots, \tilde{f}_{0, l,d})$. Each value $\tilde{f}_{0,l,i}$ corresponds to the frequency that item $i$ has in the initial local stream held by peer $l$. In this representation the operator $\oslash_p$ on a multiset translates to a multiplication of the frequencies' vector corresponding to that multiset by the scalar $1/p$. 

Now, the implicit transformation that the local streams of the selected peers, $i$ and $j$, undergo at each elementary iteration of \textsc{AVG-Merge}, i.e., eq. \eqref{elem-it-msets}, can be rewritten as:

\begin{equation}
\label{elem-it-fvecs}
\begin{split}
\boldsymbol{\tilde{F}}_r' = (\boldsymbol{\tilde{f}}_{r,1}, \boldsymbol{\tilde{f}}_{r,2}, \ldots, \frac{1}{2} (\boldsymbol{\tilde{f}}_{r,i} + \boldsymbol{\tilde{f}}_{r,j}), \ldots, \\
\frac{1}{2} (\boldsymbol{\tilde{f}}_{r,i} + \boldsymbol{\tilde{f}}_{r,j}), \ldots, \boldsymbol{\tilde{f}}_{r,p}),
\end{split}
\end{equation}

\noindent where $\boldsymbol{\tilde{F}}_r$ is a matrix whose columns are the peers' vectors of frequencies after $r$ rounds, i.e. each $\boldsymbol{\tilde{f}}_{r,l}$ is the frequencies' vector correspoding to the the multiset  ${\mathcal{N}}_{r,l}$. This matrix corresponds to the vector of multisets $\boldsymbol{\mathcal{N}}_r$ in eq. \eqref{elem-it-msets}.

Eventually, it can be recognized in eq. \eqref{elem-it-fvecs} the elementary step of the Jelasity's protocol applied in parallel to each one of the components of the frequencies' vectors of peers $i$ and $j$. We already know that the Jelasity's averaging protocol converges to the average of the values initially owned by the peers. Thus, for $r  \rightarrow \infty$,  $\boldsymbol{\tilde{F}}_r$ converges to: 

\begin{equation}
\boldsymbol{\tilde{F}}_\infty = (\boldsymbol{f}_{\rm avg}, \boldsymbol{f}_{\rm avg}, \ldots, \boldsymbol{f}_{\rm avg})
\end{equation}

\noindent where $\boldsymbol{f}_{\rm avg}$ is:

\begin{equation}
\boldsymbol{f}_{\rm avg} = (\bar{f}_1, \bar{f}_2, \ldots, \bar{f}_d),
\end{equation}

\noindent with $\bar{f}_i = \frac{1}{p} \sum_{l=1}^{p} \tilde{f}_{0,l,i}, \text{ for } i = 1,2, \ldots, d$ which is the representation as frequencies' vector of the multiset $\mathcal{N}_{\rm avg}$ in eq. \eqref{conv-msets}, proving the convergence.

\subsection{Estimating the number of peers}
As shown in the previous paragraph we need to estimate $p$, the number of peers participating in the protocol, in order to estimate the global frequencies of the items included in the final summary of a peer. 

We can do that executing in parallel with \textsc{P2PSS} an instance of the Jelasity's averaging protocol with initial values equal to $0$, except for one peer which is assigned the value $1$. In this way, the average of the values initially held by the peers is $1/p$ and we can estimate it with an error which depends on the number of rounds executed. We now analyze this error and its bound. 

According to eq. \eqref{eq_sigma2}, we have that $\sigma_0^2 = 1/p$. Let $\tilde{p}_{r,l}$ be  the estimation of the number of peers $p$ at round $r$ by the peer $l$, and $\tilde{q}_{r,l} = 1/\tilde{p}_{r,l}$. From eq. \eqref{eq_error_due_to_gossip}, it holds that, with probability $1-\delta$:
\begin{equation}
\left| \tilde{q}_{r,l} - \frac{1}{p}\right| < \sqrt{\frac{p-1}{p}} \sqrt{\frac{C^r }{\delta}} < \sqrt{\frac{C^r }{\delta}}
\end{equation}
Setting $\bar{\epsilon} = \sqrt{\frac{C^r }{\delta}}$, we have that:
\begin{equation}
\frac{1}{p} - \bar{\epsilon} < \tilde{q}_{r,l} < \frac{1}{p} + \bar{\epsilon}
\end{equation}
Assuming the constraint $\bar{\epsilon} < 1/p$, all of the members of the previous relation are positive, hence it holds that:

\begin{equation}
\label{error_bound_on_p_1}
\frac{p}{1+p\bar{\epsilon}} < \tilde{p}_{r,l} < \frac{p}{1-p\bar{\epsilon}} 
\end{equation}

The problem with eq. \eqref{error_bound_on_p_1} is that the estimation error bounds depend on $p$, but we may not know $p$ in advance. To overcome this problem, we introduce the value $p^*~\ge~p$, that is an estimate of the maximum number of peers we expect in the network, and we compute new bounds based on this value. Under the constraint $p^* \geq p$, we can be confident on the new computed bounds, though they may be weaker.

Let us set $\epsilon^* = p^* \bar{\epsilon}$. Given the constraint on $\bar{\epsilon}$, it holds that $0 < \epsilon^* < 1$, and, with probability $1-\delta$, for any peer $l = 1,2,\ldots,p$:

\begin{equation}
\label{error_bound_on_p_2}
\frac{p}{1+\epsilon^*} < \tilde{p}_{r,l} < \frac{p}{1-\epsilon^*} 
\end{equation}

\subsection{Gossip-based approximation}
In the discussion on the convergence of \textsc{P2PSS}, we have seen that, at round $r$ and for a peer $l$, the summary $\mathcal{S}_{r,l}$ held by that peer implicitly refers to a stream represented by the multiset $\mathcal{N}_{r,l}$, or the frequencies' vector $\boldsymbol{\tilde{f}}_{r,l}$. Thus, the eqs. \eqref{ssprop1}--\eqref{ssprop4} are valid for $\mathcal{S}_{r,l}$ with reference to $\mathcal{N}_{r,l}$. As a consequence, we need to compute how far the frequencies of items in $\boldsymbol{f}_{r,l}$ are from those in $\boldsymbol{f}_{\rm avg}$, that is the vector of true average frequencies.

For what we said in the previous paragraph we can do that by referring to the Jelasity's protocol and eq. \eqref{eq_error_due_to_gossip}. Let us denote by $f_i$ the global frequency of item $i$ and let $\tilde{f}_{r,l,i}$ be the estimation of the average frequency of that item, i.e. $f_i/p$ by peer $l$, after round $r$. According to eq. \eqref{eq_error_due_to_gossip}, with probability $1-\delta$ for any peer $l\in[p]$ and any item $i\in[d]$:

\begin{equation}
\left|\tilde{f}_{r,l,i} - \frac{f_i}{p}\right| < \sqrt{(p-1)\sigma_0^2} \sqrt{\frac{C^r }{\delta}}
\end{equation}

The initial distribution $\sigma_0^2$ of the local frequencies of the chosen item over the peers is not known in advance, but the worst case happens when only one peer has the whole quantity $f_i$ and the other $p-1$ peers hold the value $0$. In this case, it follows that  $\sigma_0^2 \leq f_i^2/p$, and hence, with probability $1-\delta$:

\begin{equation}
\left|\tilde{f}_{r,l,i} - \frac{f_i}{p}\right| < f_i\sqrt{\frac{p-1}{p}} \sqrt{\frac{C^r }{\delta}} < f_i \sqrt{\frac{C^r }{\delta}}.
\end{equation}

Considering the definition of $\bar{\epsilon}$ and $\epsilon^*$, it holds that:

\begin{equation}
\frac{f_i}{p} - f_i\bar{\epsilon} < \tilde{f}_{r,l,i} < \frac{f_i}{p} + f_i\bar{\epsilon}
\end{equation}

\noindent that is:

\begin{equation}
\frac{f_i}{p} (1 - \epsilon^*) < \tilde{f}_{r,l,i} < \frac{f_i}{p} (1 + \epsilon^*)
\end{equation}

With a similar reasoning, we can also determine a relationship between the sum of all of the local items' frequencies, for any peer $l$, after the $r$-th round of the algorithm, i.e., $\tilde{n}_{r,l} = |\mathcal{N}_{r,l}|$,  and the sum of all of the items' frequencies in the global stream, $n = |\mathcal{N}|$. With probability $1-\delta$, for any peer $l\in[p]$ and any item $i\in[d]$:

\begin{equation}
\frac{n}{p} (1 - \epsilon^*) < \tilde{n}_{r,l} < \frac{n}{p} (1 + \epsilon^*).
\end{equation}

\subsection{Space-Saving approximation}
At last, let us consider again the invariants of our algorithm: after a round of \textsc{P2PSS}, the summary held by a peer changes and the local stream to which that peer refers changes accordingly so that each peer continues to hold a correct summary for its corresponding portion of the input global stream. This means that each peer's summary $\mathcal{S}_{r,l}$ estimates the frequency of an item in the redistributed local stream $\mathcal{N}_{r,l}$ within the error bounds guaranteed by eqs. \eqref{ssprop1}--\eqref{ssprop4}.
Consequently, denoting by $\hat{f}_{r,l,i}$ the frequency of an item $i$ in $\mathcal{S}_{r,l}$ and by $f_{r,l,i}$ the frequency of that item in $\mathcal{N}_{r,l}$, we have that, for any peer $l\in[p]$ and any item $i\in[d]$:

\begin{equation}
\tilde{f}_{r,l,i} \le \hat{f}_{r,l,i} \le \tilde{f}_{r,l,i} + \frac{\tilde{n}_{r,l}}{k}
\end{equation}

\subsection{Correctness and error bounds}
We shall show here that given a summary $\mathcal{S}_{r,l}$ obtained by any peer $l$ after $r$ rounds of \textsc{P2PSS}, we can select a set of items and their corresponding estimated frequencies solving the Approximate Frequent Items Problem in Unstructured P2P Networks stated in Section~\ref{preliminary-defs}. We shall also determine the error bounds on frequencies' estimation and the relation among the number $k$ of counters to be used by each node and the number $r$ of rounds to be executed in order to guarantee the false positives' tolerance requested by the user.  

\begin{thm}
\label{thm_correctness}
	Given an input stream $\mathcal{N}$ of length $n$, distributed among $p$ nodes, a threshold parameter $0~<~\phi~<~1$, and a probability of failure $0 < \delta < 1$, after $r$ rounds of \textsc{P2PSS}, any peer can returns a set $H$ of items and their corresponding estimated frequencies, so that, with probability $1-\delta$:
	\begin{enumerate}
		\item $H$ includes all of the items in $\mathcal{N}$ that have frequency $f > \phi n$;
		\item $H$ does not include any items in $\mathcal{N}$ that have frequency $f \le (\phi - \epsilon) n$;
	\end{enumerate}
with a false positives tolerance $\epsilon = \frac{4\epsilon^* \phi}{(1+\epsilon^*)^2} + \frac{1-\epsilon^*}{k (1+\epsilon^*)}$ which is bonded by the number of counters $k$ used for the summaries and the number of rounds $r$ executed.
\end{thm}

\begin{proof}
We first recap the main relations we proved above, valid with probability $1-\delta$, for all the items $i$ in the summary $\mathcal{S}_{r,l}$ and any given peer $l$, after round $r$:

\begin{align}
\label{eq_properties1}
\frac{p}{1+\epsilon^*} < & \tilde{p}_{r,l} < \frac{p}{1-\epsilon^*}; \\
\label{eq_properties2}
\frac{f_i}{p} (1 - \epsilon^*) < & \tilde{f}_{r,l,i} < \frac{f_i}{p} (1 + \epsilon^*); \\
\label{eq_properties3}
\frac{n}{p} (1 - \epsilon^*) < & \tilde{n}_{r,l} < \frac{n}{p} (1 + \epsilon^*); \\
\label{eq_properties4}
\tilde{f}_{r,l,i} \le & \hat{f}_{r,l,i} \le \tilde{f}_{r,l,i} + \frac{\tilde{n}_{r,l}}{k};
\end{align}

We need to select all of the items whose global frequency $f_i$ is greater than the threshold $\phi n$. From the relations \eqref{eq_properties1}--\eqref{eq_properties4}, we can derive the following:

\begin{equation}
\hat{f}_{r,l,i} \frac{p}{1-\epsilon^*} > \tilde{f}_{r,l,i} \frac{p}{1-\epsilon^*} > f_i > \phi n > \phi \tilde{n}_{r,l} \frac{p}{1+\epsilon^*}
\end{equation}

Thus, we do not need to output all of the items in the summary $\mathcal{S}_{r,l}$, but only those ones which have an estimated frequency respecting the following relation:

\begin{equation}
\hat{f}_{r,l,i} > \phi \tilde{n}_{r,l} \frac{1-\epsilon^*}{1+\epsilon^*}
\end{equation}

In order to compute the error, in terms of false positives' tolerance, that we commit with this selection criterion, we can use again eqs. \eqref{eq_properties1}--\eqref{eq_properties4} and prove that if $\hat{f}_{r,l,i} > \phi \tilde{n}_{r,l} \frac{1-\epsilon^*}{1+\epsilon^*}$, then, with probability $1-\delta$:

\begin{equation}
f_i > \left\{\phi - \left[\frac{4\epsilon^* \phi}{(1+\epsilon^*)^2} + \frac{1}{k}\right]\right\} n.
\end{equation}

\noindent In fact:


\begin{equation}
\begin{split}
\tilde{f}_{r,l,i} + \frac{\tilde{n}_{r,l}}{k} > \hat{f}_{r,l,i} > \phi \tilde{n}_{r,l} \frac{1-\epsilon^*}{1+\epsilon^*}  \implies \\
\frac{\tilde{f}_{r,l,i}}{\tilde{n}_{r,l}} + \frac{1}{k} > \phi \frac{1-\epsilon^*}{1+\epsilon^*}  \implies \\
\frac{f_i(1+\epsilon^*)}{n(1-\epsilon^*)} + \frac{1}{k}> \phi \frac{1-\epsilon^*}{1+\epsilon^*} \implies \\
f_i > \phi n \left(\frac{1-\epsilon^*}{1+\epsilon^*}\right)^2 - \frac{n(1-\epsilon^*)}{k(1+\epsilon^*)}\implies \\
f_i > \left\{\phi - \left[1 - \left(\frac{1-\epsilon^*}{1+\epsilon^*}\right)^2\right] \phi + \frac{1-\epsilon^*}{k(1+\epsilon^*)}\right\}  n \implies \\
f_i > \left\{\phi - \left[\frac{4\epsilon^* \phi}{(1+\epsilon^*)^2} + \frac{1-\epsilon^*}{k(1+\epsilon^*)}\right]\right\} n
\end{split}
\end{equation}

Thus, we can conclude that, with reference to the problem definition, with probability $1-\delta$, no items with frequency $f_i \le (\phi - \epsilon) n$ shall be reported in $H$, with $\epsilon = \frac{4\epsilon^* \phi}{(1+\epsilon^*)^2} + \frac{1-\epsilon^*}{k (1+\epsilon^*)}$.
\end{proof}

\subsection{Frequency estimation error bounds}
The frequency estimations in $\mathcal{S}_{r,l}$ are referred to average frequencies. Thus, in order to obtain an estimation of the global frequency $f_i$ of an item $i$, we need to multiply $\hat{f}_{r,l,i}$ by $\tilde{p}_{r,l}$.
From eqs. \eqref{eq_properties1}--\eqref{eq_properties4} we can compute the error bounds of this estimation. The following theorem holds.

\begin{thm}
Given an input stream $\mathcal{N}$ of length $n$, distributed among $p$ nodes and a probability of failure $0< \delta < 1$, after $r$ rounds of \textsc{P2PSS}, any peer can report a frequency estimation $f^s_{r,l,i}$ of an item $i \in [m]$ so that, with probability $1-\delta$:

	\begin{equation}
		\label{est_error}
		\frac{1-\epsilon^*}{1+\epsilon^*} f_i < f^s_{r,l,i} < \frac{1+\epsilon^*}{1-\epsilon^*} \left(f_i + \frac{n}{k}\right).
	\end{equation}
\end{thm}

\begin{proof}
	From eq. \eqref{eq_properties1} we have that:
	\begin{equation}
	 	\label{eq_properties5}
		 \frac{1}{1+\epsilon^*} < \frac{\tilde{p}_{r,l}}{p} < \frac{1}{1-\epsilon^*} 
	\end{equation}
\noindent and from eq. \eqref{eq_properties2} and eq. \eqref{eq_properties3}, we have that:
	\begin{equation}
	   \label{eq_properties6}
		\begin{split}
			f_i \frac{\tilde{p}_{r,l}}{p} (1 - \epsilon^*) < \tilde{f}_{r,l,i}  \tilde{p}_{r,l} < f_i \frac{\tilde{p}_{r,l}}{p} (1 + \epsilon^*), \\
			n \frac{\tilde{p}_{r,l}}{p} (1 - \epsilon^*) < \tilde{n}_{r,l} \tilde{p}_{r,l} < n \frac{\tilde{p}_{r,l}}{p} (1 + \epsilon^*). \\
		\end{split}
	\end{equation}
\noindent Now, starting from eq. \eqref{eq_properties4} and taking into account eq. \eqref{eq_properties5} and eq. \eqref{eq_properties6}, it follows that:
	\begin{equation}
		\begin{split}
			\tilde{f}_{r,l,i} \tilde{p}_{r,l} \le &\hat{f}_{r,l,i} \tilde{p}_{r,l} \le \tilde{f}_{r,l,i} \tilde{p}_{r,l} + \frac{\tilde{n}_{r,l}}{k} \tilde{p}_{r,l} \implies \\
			f_i \frac{\tilde{p}_{r,l}}{p} (1 - \epsilon^*) < &\hat{f}_{r,l,i} \tilde{p}_{r,l} < \left(f_i + \frac{n}{k}\right) \frac{\tilde{p}_{r,l}}{p} (1 + \epsilon^*) \implies \\
				\frac{1-\epsilon^*}{1+\epsilon^*} f_i < &\hat{f}_{r,l,i} \tilde{p}_{r,l}  < \frac{1+\epsilon^*}{1-\epsilon^*} \left(f_i + \frac{n}{k}\right).
		\end{split}
	\end{equation}
\noindent and eventually, setting $f^s_{r,l,i} = \hat{f}_{r,l,i} \tilde{p}_{r,l}$, the relation \eqref{est_error} follows.
\end{proof}

\subsection{Practical considerations}

We conclude this Section discussing how to select proper values for the parameters $k$ and $R$, which represent respectively the number of counters to be used for the Space-Saving stream summary data structure and the minimum number of rounds required to solve the Approximate Frequent items Problem in Unstructured P2P Networks. Theorem~\ref{thm_correctness} proves the correctness of the algorithm providing also a theoretical guarantee about the bound $\epsilon$ on the number of false positives items. The user can increase the number of rounds $R$ and/or increase the number of Space-Saving counters $k$ to reduce the  false positives tolerance $\epsilon$. Fixing a given tolerance $\epsilon$, the user has one degree of freedom to achieve it; Figure~\ref{plot_k_r} plots the relationship between the values for $R$ and $k$ which produce a given tolerance $\epsilon$. The relationship between $k$ and $R$ is given by eq. \eqref{eq_k_r}:

\begin{equation}
 \label{eq_k_r}
k=\frac{1-\epsilon^{*^2}}{\epsilon \left(1+\epsilon^*\right)^2-4 \phi \epsilon^*}=\frac{1-p^{*^2} \frac{C^R}{\delta}}{\epsilon \left(1+p^* \sqrt{\frac{C^R}{\delta}}\right)^2-4 \phi p^* \sqrt{\frac{C^R}{\delta}}}
\end{equation}

\begin{figure}[]
  \centering
	\includegraphics[width=0.45\textwidth]{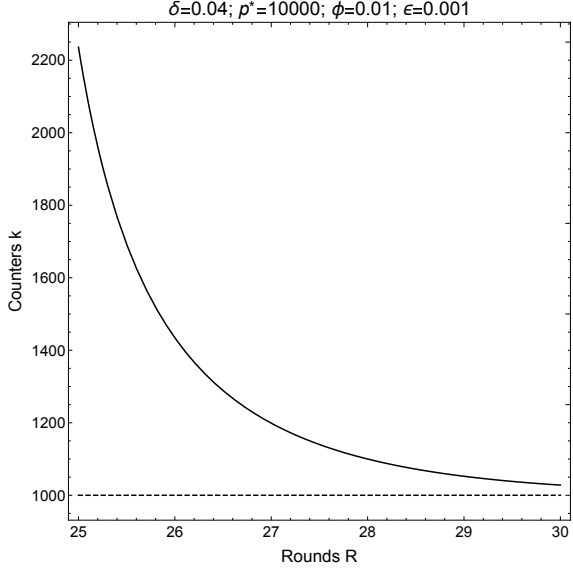}
 	\caption{Relationship between the number of counters and the number of rounds to guarantee a given level of false positive tolerance $\epsilon$.} 
  	\label{plot_k_r}
\end{figure}

Among all of the possible values for $R$ and $k$, the user could follow a strategy oriented to maintain the number of rounds (hence the time) as fewer as possibile and to choose $k$ accordingly or viceversa to maintain the number of counters (hence the space) as lower as possibile and to choose $R$ accordingly. Let us now discuss both strategies.

With the first strategy, which can be called \textit{time-dominant}, the user is interested on choosing $R$ and $k$ which guarantee a given $\epsilon$ such that $R$ is minimum. The eq. \eqref{eq_k_r} reveals that $R$ is a monotone decreasing function with $k$, hence the minimum value for $R$ is obtained when $k$ tends to infinity; moreover, it holds that $k>0$ hence the minimum value for $R$ can be calculated by imposing the following constraint:
 
\begin{equation}
\epsilon \left(p^* \sqrt{\frac{C^R}{\delta}}+1\right)^2-4 \phi p^* \sqrt{\frac{C^R}{\delta}} > 0
\end{equation}

\noindent from which it follows that

\begin{equation}
R> \frac{ \log {\delta}+2 \log \left( \frac{2 \phi -\epsilon - 2 \sqrt{\phi^2-\epsilon \phi}}{\epsilon p^*} \right)}{\log C}
\end{equation}

\noindent Since $R$ is an integer, the minimum value of $R$ is given by:

\begin{equation}
\label{eq_min_r}
R_{min}=\left\lfloor \frac{ \log {\delta}+2 \log \left( \frac{2 \phi -\epsilon - 2 \sqrt{\phi^2-\epsilon \phi}}{\epsilon p^*} \right)}{\log C}  \right\rfloor +1
\end{equation}

\noindent Substituting the vale of $R_{min}$ provided by eq. \eqref{eq_min_r} into eq. \eqref{eq_k_r} for $R$, it is possible to obtain the value for $k$.

With the second strategy, which can be called \textit{space-dominant}, the user is interested to keeping the memory footprint as lower as possibile. The eq. \eqref{eq_k_r} reveals that $k$ is a monotone decreasing function with $R$ hence the minimum value for $k$ is obtained when $R$ tends to infinity. Evaluating eq. \eqref{eq_k_r} for $R \rightarrow \infty$ it holds that the minimum value for $k$ is given by:

\begin{equation}
k>\frac{1}{\epsilon}.
\end{equation}

\noindent Considering that $k$ is an integer value

\begin{equation}
\label{eq_min_k}
k_{min}= \left \lfloor \frac{1}{\epsilon} \right \rfloor + 1
\end{equation}

\noindent solving eq. \eqref{eq_k_r} by $\epsilon^*$ and using eq. \eqref{eq_min_k} it holds that:

\begin{equation}
\epsilon^* = \frac{k_{min}(2\phi -\epsilon) - \sqrt{4 \phi k^2_{min} (\phi -\epsilon)+ 1}}{1 + \epsilon k_{min}}.
\end{equation}

\noindent Since $\epsilon^* = p^* \sqrt{\frac{C^R}{\delta}}$, it holds that:

\begin{equation}
\label{eq_r_min_k}
R = \frac{1}{\log C} \left( 2\log \epsilon^* - 2 \log p^* + \log \delta \right).
\end{equation}

Since $R$ is an integer, 

\begin{equation}
\label{eq_r_min_k}
R = \left \lfloor \frac{1}{\log C} \left( 2\log \epsilon^* - 2 \log p^* + \log \delta \right) \right \rfloor + 1.
\end{equation}

\section{Experimental results}
\label{results}
In order to  evaluate our \textsc{P2PSS} algorithm we have implemented a simulator in C++ using the igraph library \cite{libigraph}, and carried out a series of experiments. The simulator has been compiled using the GNU C++ compiler g++ 4.8.5 on CentOS Linux 7. The tests have been performed on a machine equipped with two hexa-core Intel Xeon-E5 2620 CPUs at 2.0 GHz and 64 GB of main memory.
The source code of the simulator is freely available for inspection and for reproducibility of results contacting the authors by email. 

In every experiment, a global input stream of items has been generated (items are 32 bits unsigned integers, but the source code implementing the algorithm can be easily modified in order to process different types of items) following a Zipfian distribution and each peer has been assigned a distinct part of that global stream, thus simulating the scenario in which each peer processes, independently of the other peers, its own local sub-stream, and the peers collaboratively discover the frequent items in the union of their sub-streams. The experiments have been repeated 10 times setting each time a different seed for the pseudo-random number generator used for creating the input data. For each experiment execution, we collected the peers' statistics relevant for the evaluation of the algorithm (more details in the following). Then, with reference to each peer, we determined the average value of those statistics over the ten executions. At last, we computed the mean and confidence interval for each statistics over all of the peers and plotted this values.

We fixed the number of elements in the global stream at $200$ millions, and varied the skew of the Zipfian distribution, $\rho$, the number of peers, $p$, the frequent items threshold, $\phi$, the number of counters used by each peer $k$ or the fan-out $fo$, setting non varying parameters to the default values. Every experiment has been carried out by generating random P2P network topologies through the Barabasi-Albert and Erdos-Renyi random graphs models. Table~\ref{experiments} reports the sets of values (first row) and default values (second row) used for the parameters.

\begin{table*}
	\renewcommand{\arraystretch}{1.3}
	\caption{Experiment values}
	\label{experiments}
	\centering
	\small
	\begin{tabular}{| c | c | c | c | c | c |}
		\hline
		$\boldsymbol{\rho}$ &  $\boldsymbol{\phi}$ & $\boldsymbol{p\  (\times10^3)}$  & $\boldsymbol{k\  (\times10^3)}$ & $\boldsymbol{r}$ & $\boldsymbol{fo}$\\
		\hline
		 \{0.9, 1.1, 1.3, 1.5\} & \{0.01, 0.02, 0.03, 0.04\}  & \{1, 5, 10, 15, 20\} &  \{1, 1.8, 2.6, 3.4\} &  \{20, 22, 24, 26, 28\} & \{1, 2, 3, ALL\}\\ 
		\hline
		1,2 & 0.02 & 10 &  2.2 &  24 &  1\\ \hline
	\end{tabular}
\end{table*}

The metrics computed are the \emph{Recall}, the \emph{Precision}, and the \emph{Average Relative Error} on frequency estimation with reference to the set of frequent items candidates reported in output. Recall is defined as the fraction of frequent items retrieved by an algorithm over the total number of frequent items.  Precision is the fraction of frequent items retrieved over the total number of items reported as frequent items candidates. Relative Error is defined as usual as $\frac{\left| f^S-f \right|}{f}$, where $f^S$ is the frequency reported for an item and $f$ is its true frequency. 

\begin{figure*}[h]
	\centering
	\begin{tabular}{ccc}		
		\subfloat[Recall]{
			\includegraphics[width=0.3\textwidth]{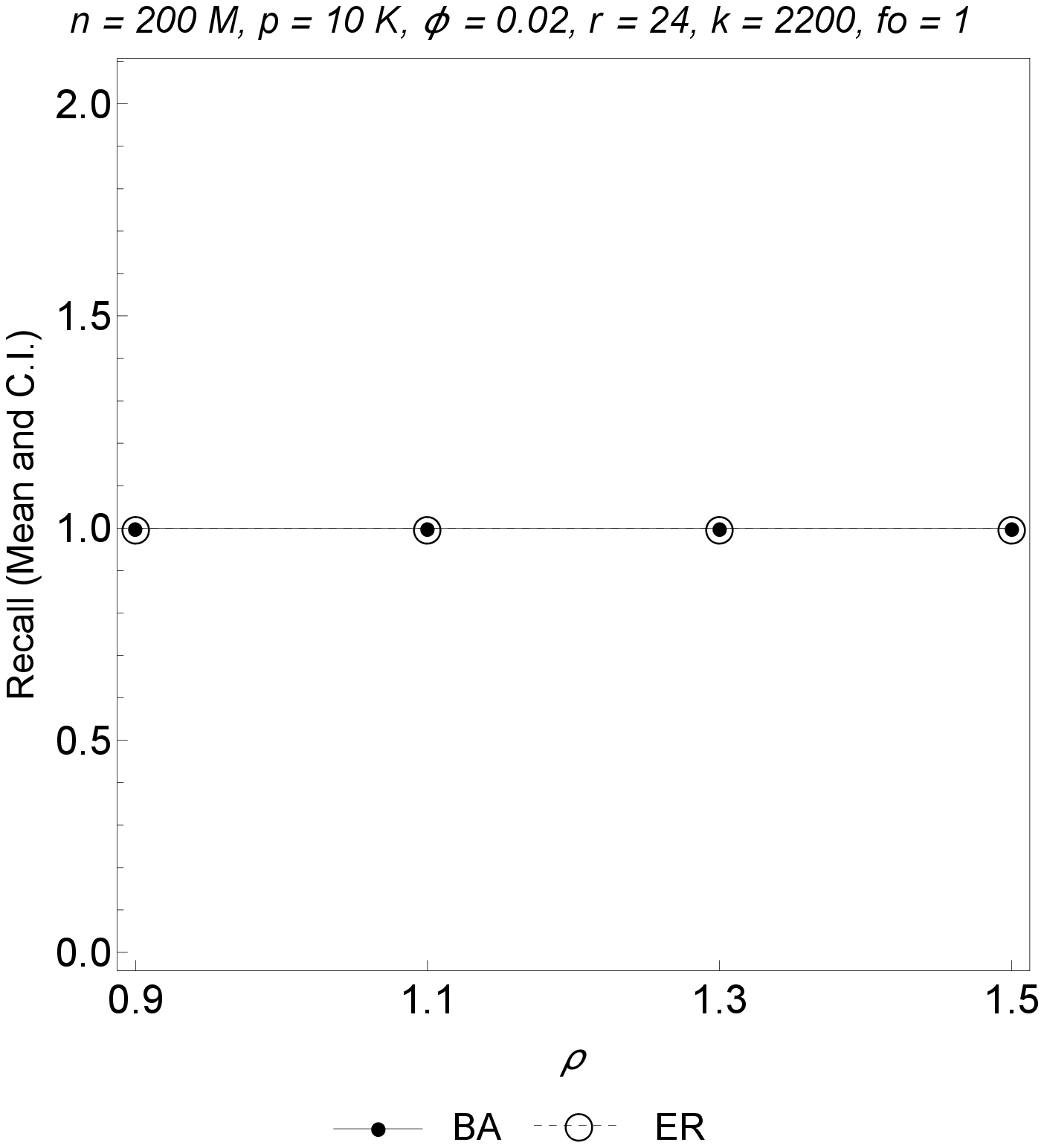}
			\label{sk-rec}
		} &
		
		\subfloat[Precision]{
			\includegraphics[width=0.3\textwidth]{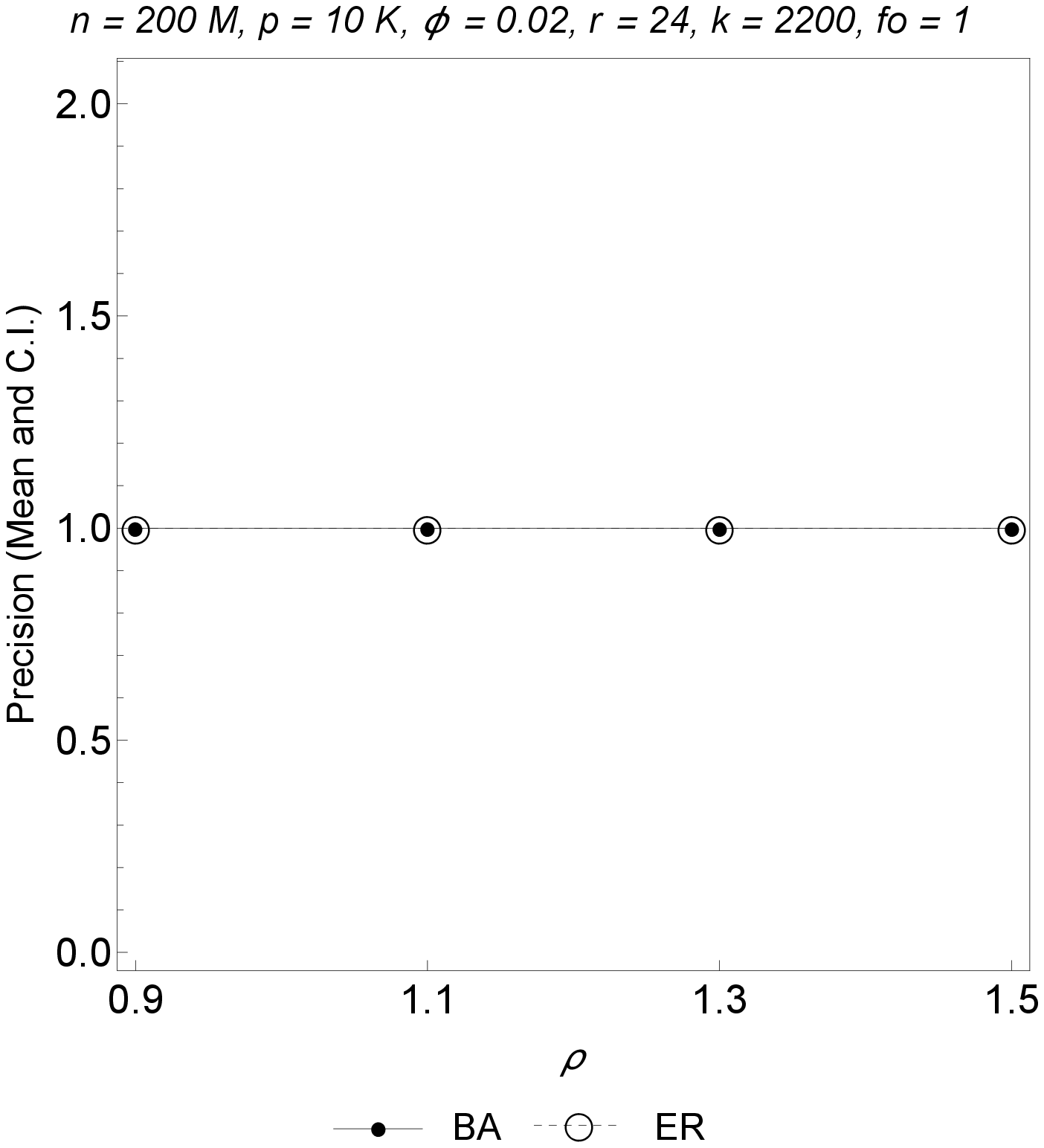}
			\label{sk-prec}
		} &
	  	
			\subfloat[Average Relative Error ]{
			\includegraphics[width=0.3\textwidth]{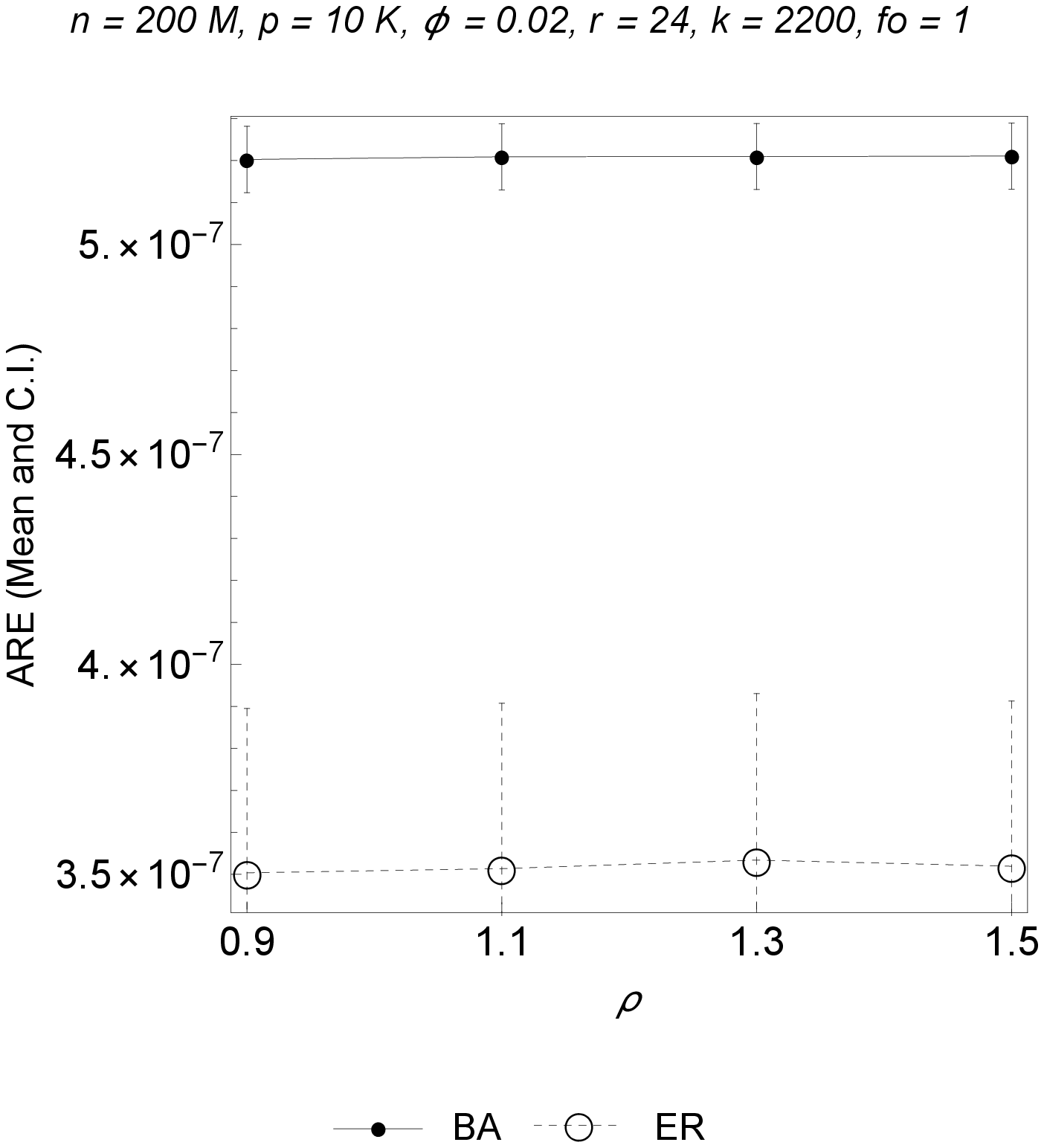}
			\label{sk-are}
		} 
	\end{tabular}
	
	\caption{Recall, Precision and Average Relative Error (mean and confidence interval) varying the skewness of the input distribution, for both a Barabasi-Albert (BA) and an Erdos-Renyi (ER) type of network graph.} 
	\label{skew_plot}
\end{figure*}

Figure~\ref{skew_plot} reports the Recall (Fig.~\ref{sk-rec}), the Precision (Fig.~\ref{sk-prec}) and the Average Relative Error (Fig.~\ref{sk-are}) varying the skewness of the Zipfian disribution from which the input items are drawn. Recall and Precision are always $100\%$, showing that the algorithm is robust enough with regard to skewness variations in the input. Moreover, Average Relative Errors on frequency estimation are very low, and in particular we note that an increase in the fan-out from 1 to 2 improves the accuracy of estimation.

\begin{figure*}[h]
	\centering
	\begin{tabular}{ccc}		
		\subfloat[Recall]{
			\includegraphics[width=0.3\textwidth]{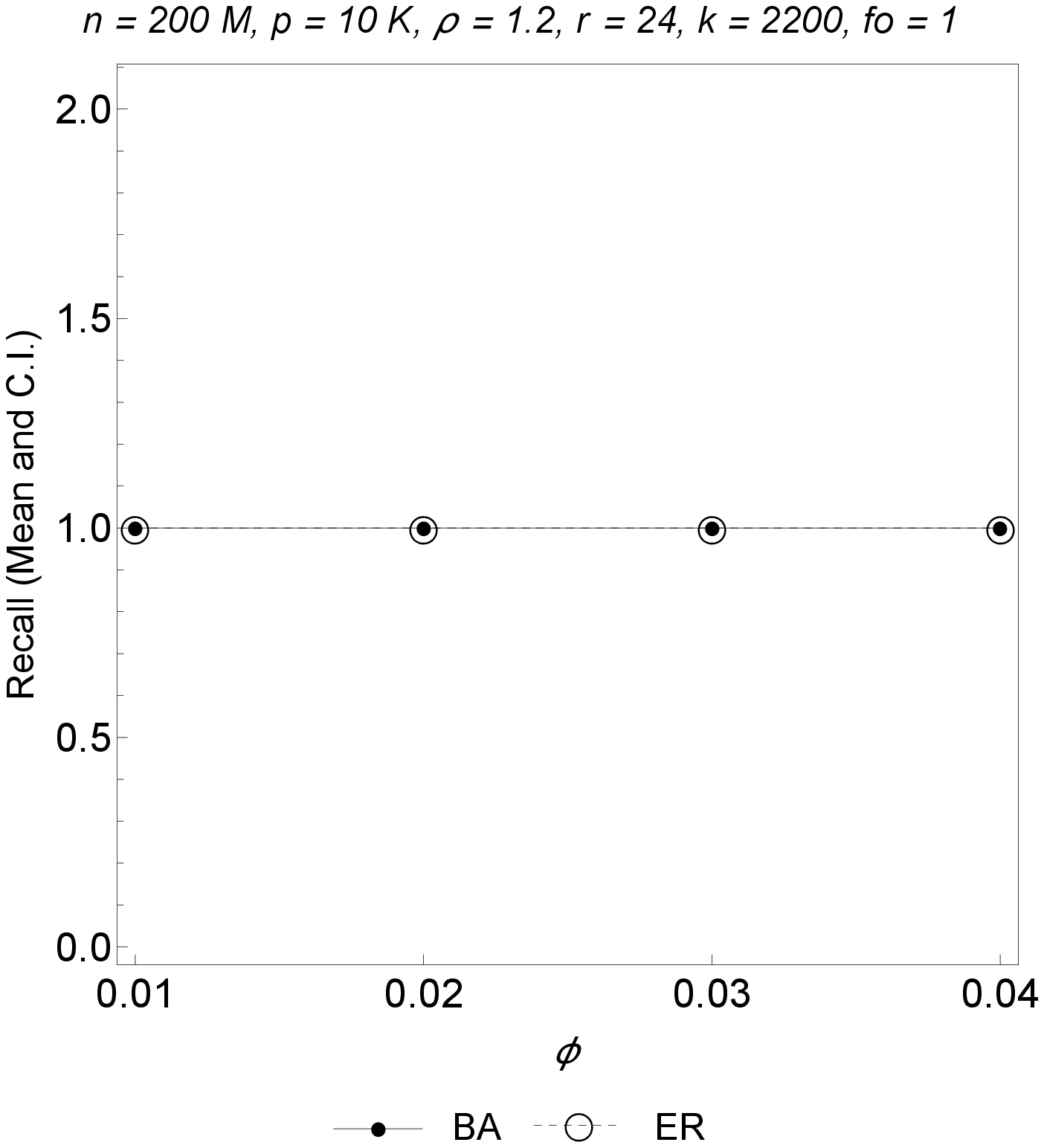}
			\label{phi-rec}
		} &
		
		\subfloat[Precision]{
			\includegraphics[width=0.3\textwidth]{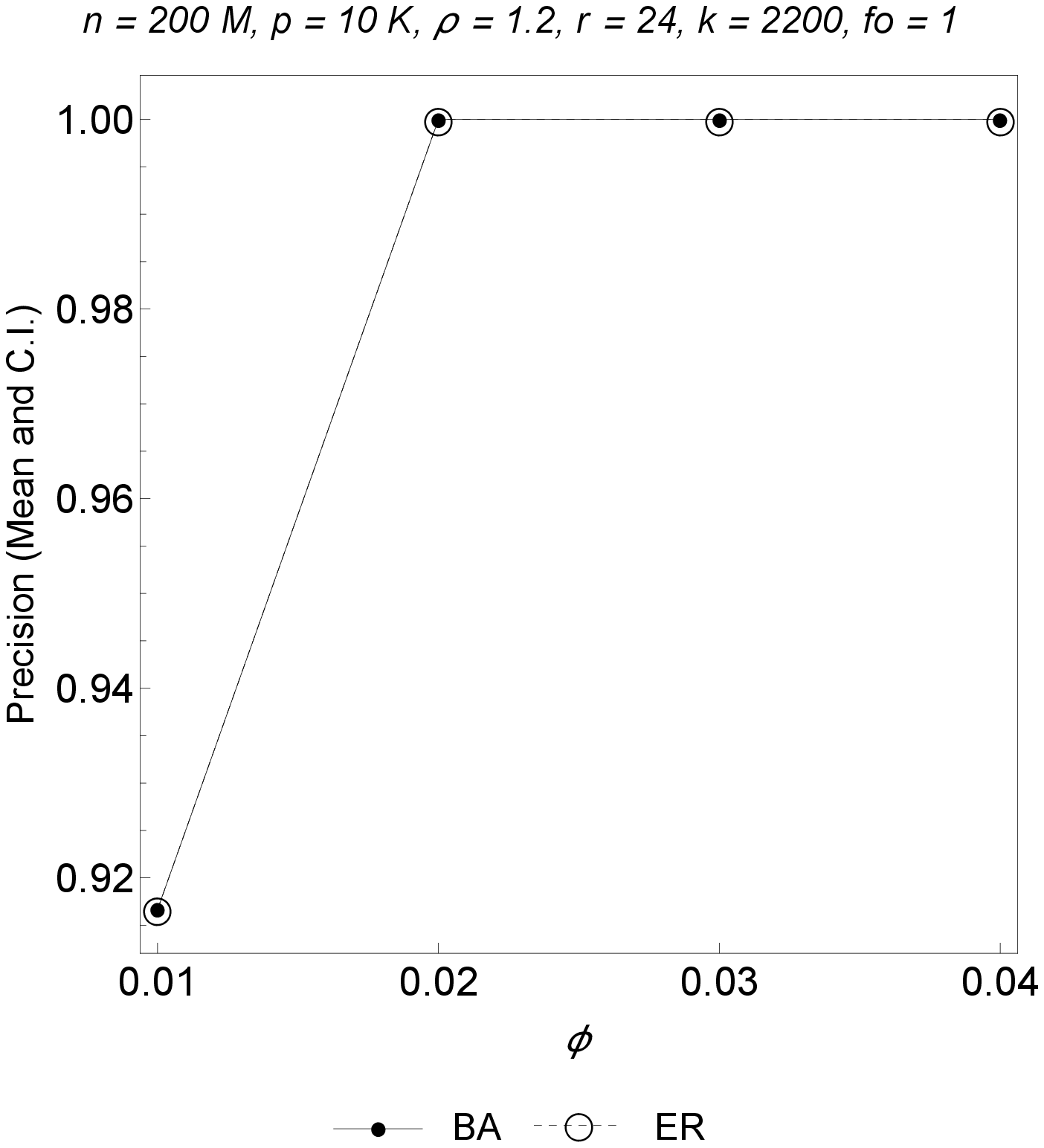}
			\label{phi-prec}
		} &
		
			\subfloat[Average Relative Error ]{
				\includegraphics[width=0.3\textwidth]{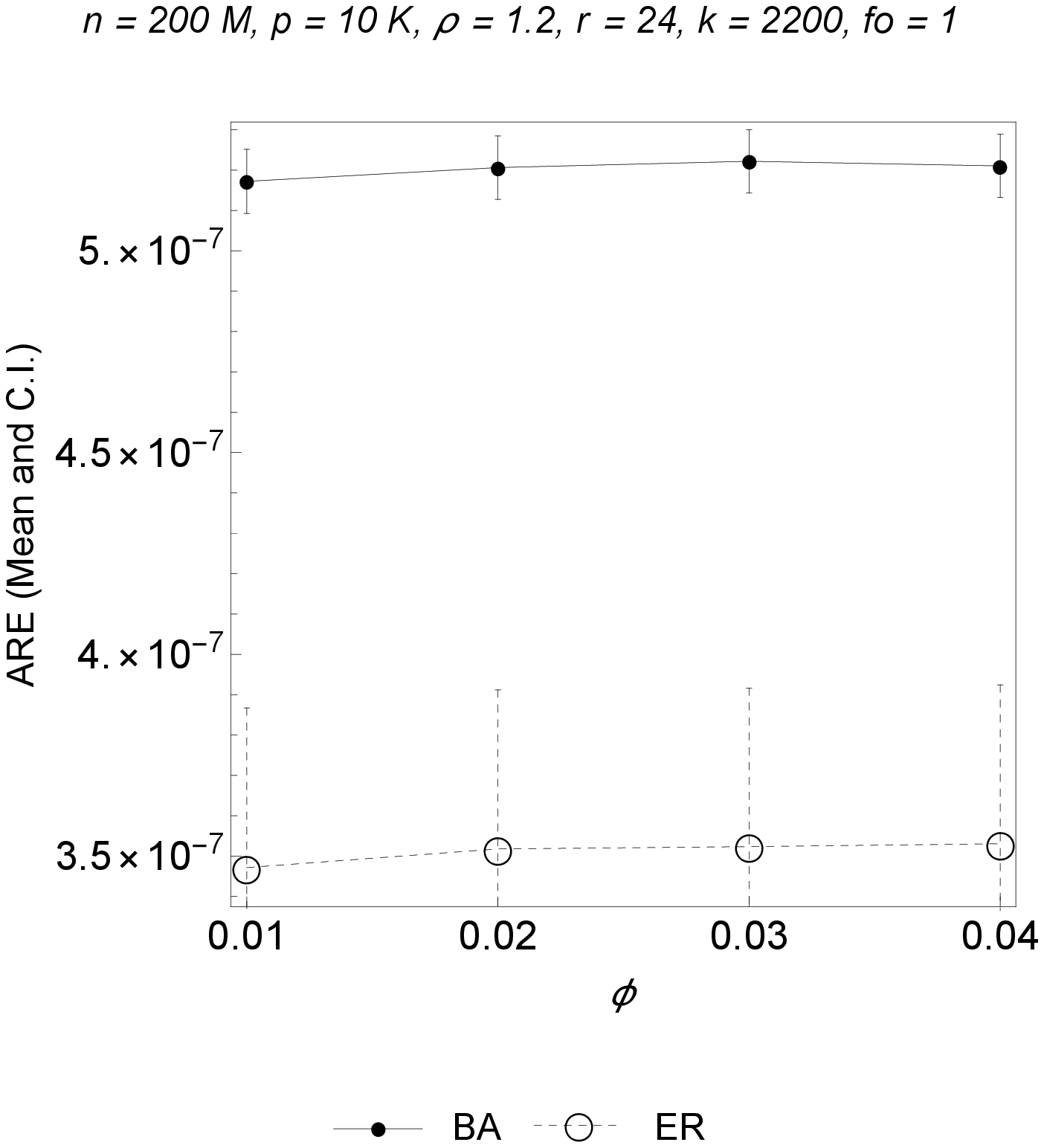}
				\label{phi-are}
			} 
		
	\end{tabular}
	
	\caption{Recall, Precision and Average Relative Error (mean and confidence interval) varying the frequent items threshold $\phi$, for both a Barabasi-Albert (BA) and an Erdos-Renyi (ER) type of network graph.} 
	\label{phi_plot}
\end{figure*}

Figure~\ref{phi_plot} shows how  \textsc{P2PSS} behaves with regard  to variations of the threshold $\phi$. The figure confirms a good performance of the algorithm: Recall is always $100\%$ as well as the Precision, except for a slightly lower value for $\phi = 0.01$. Average Relative Errors are at the same levels as for the skewness plots.

\begin{figure*}[h]
	\centering
	\begin{tabular}{ccc}		
		\subfloat[Recall]{
			\includegraphics[width=0.285\textwidth]{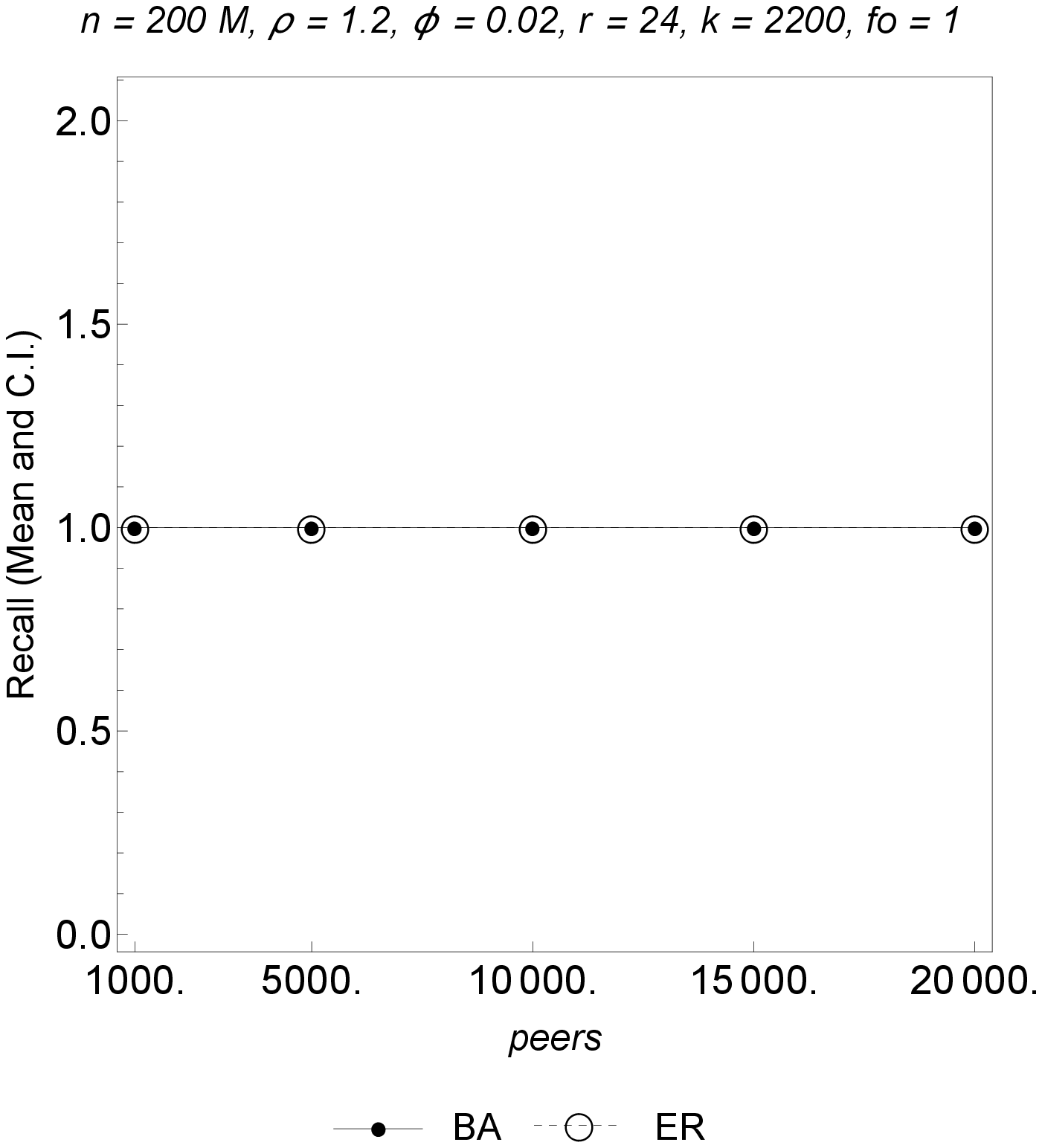}
			\label{p-rec}
		} &
		
		\subfloat[Precision]{
			\includegraphics[width=0.285\textwidth]{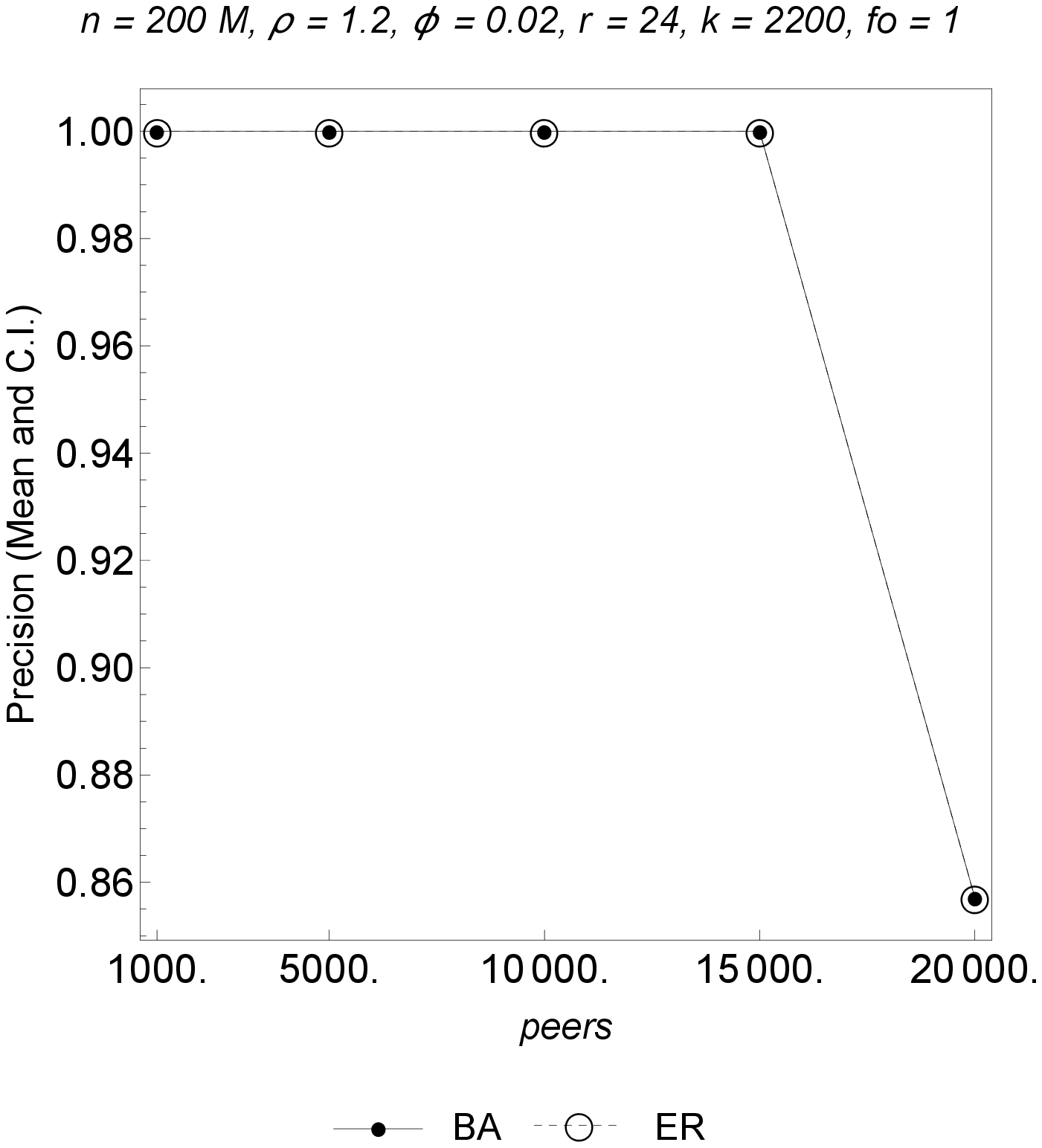}
			\label{p-prec}
		} &
		
			\subfloat[Average Relative Error ]{
				\includegraphics[width=0.33\textwidth]{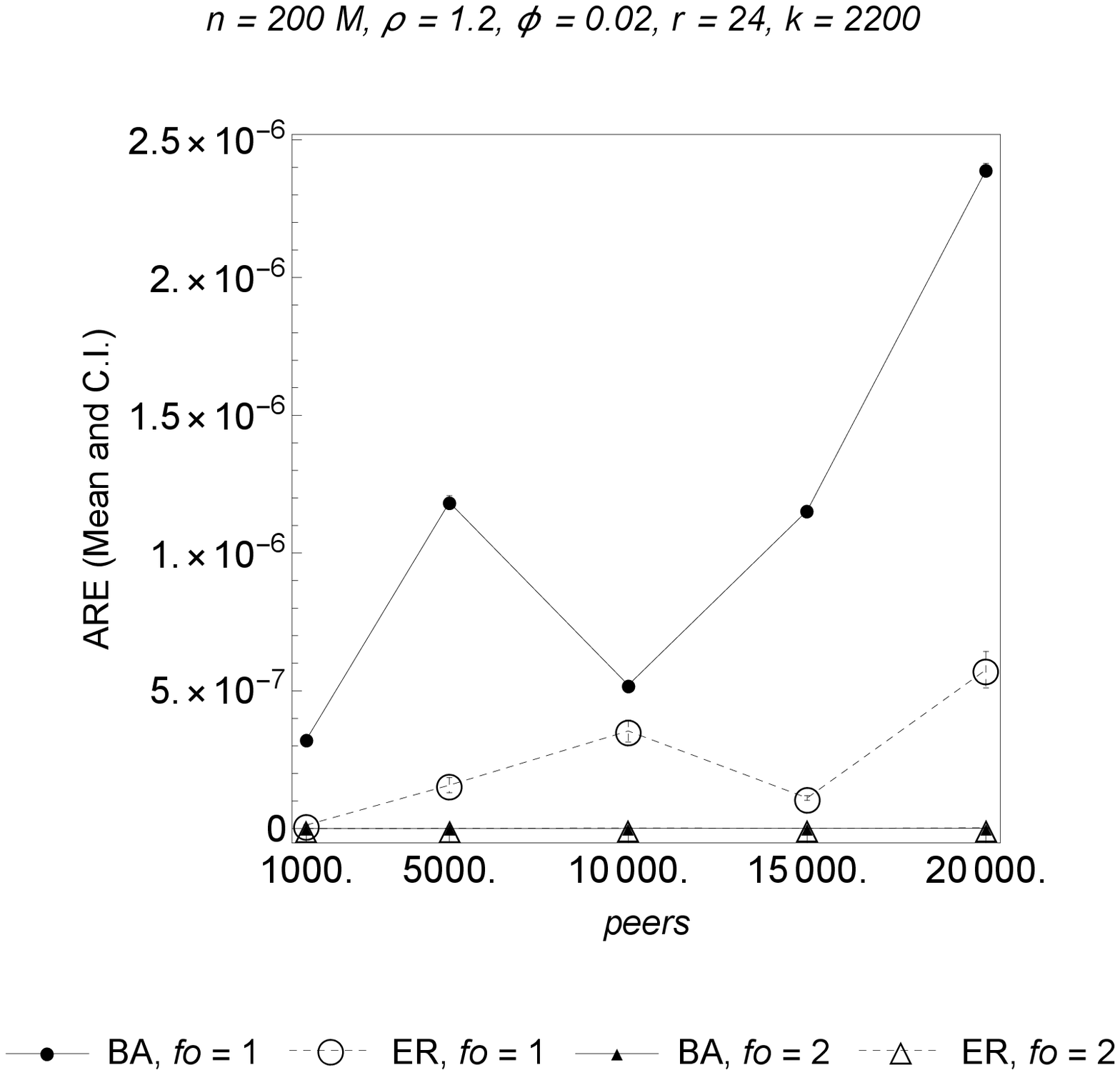}
				\label{p-are}
			}
	
	\end{tabular}
	
	\caption{Recall, Precision and Average Relative Error (mean and confidence interval) varying the number of peers participating in the computation,  for both a Barabasi-Albert (BA) and an Erdos-Renyi (ER) type of network graph, and setting a fan-out $fo$ equal to 1 and 2 in case of the ARE plot.} 
	\label{peers_plot}
\end{figure*}

Figure~\ref{peers_plot} depicts the trend for Recall, Precision and Average Relative Error with regard to the experiments where we varied the number of peers. As we expect from the theoretical analysis, here the Precision suffers a reduction and the Average Relative Error increases when the number of peers grows too much respect to the number of counters used (the default value is fixed to $2200$) and the number of rounds executed (the default value is fixed to $24$).

\begin{figure*}[h]
	\centering
	\begin{tabular}{ccc}		
		\subfloat[Recall]{
			\includegraphics[width=0.3\textwidth]{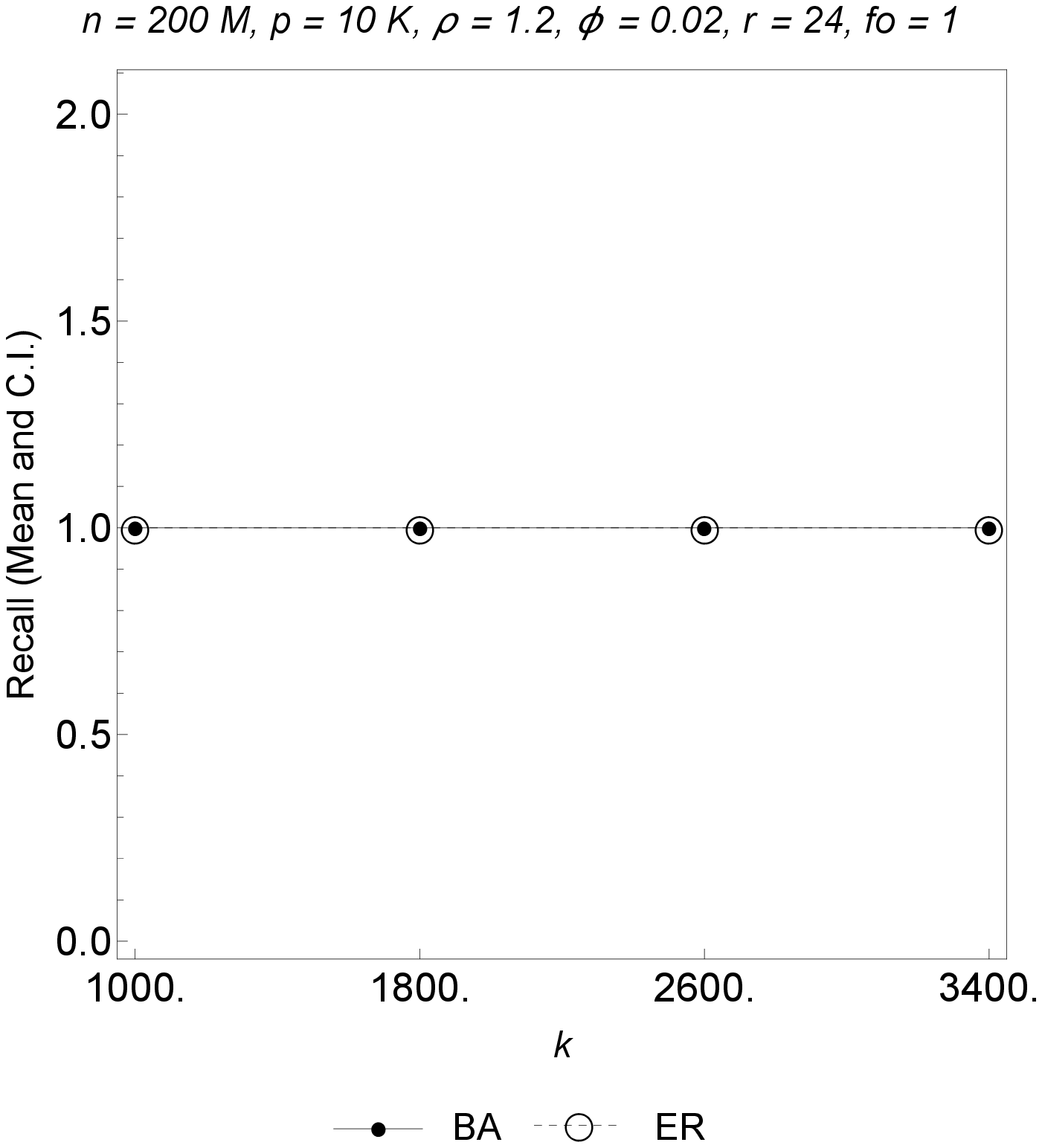}
			\label{k-rec}
		} &
		
		\subfloat[Precision]{
			\includegraphics[width=0.3\textwidth]{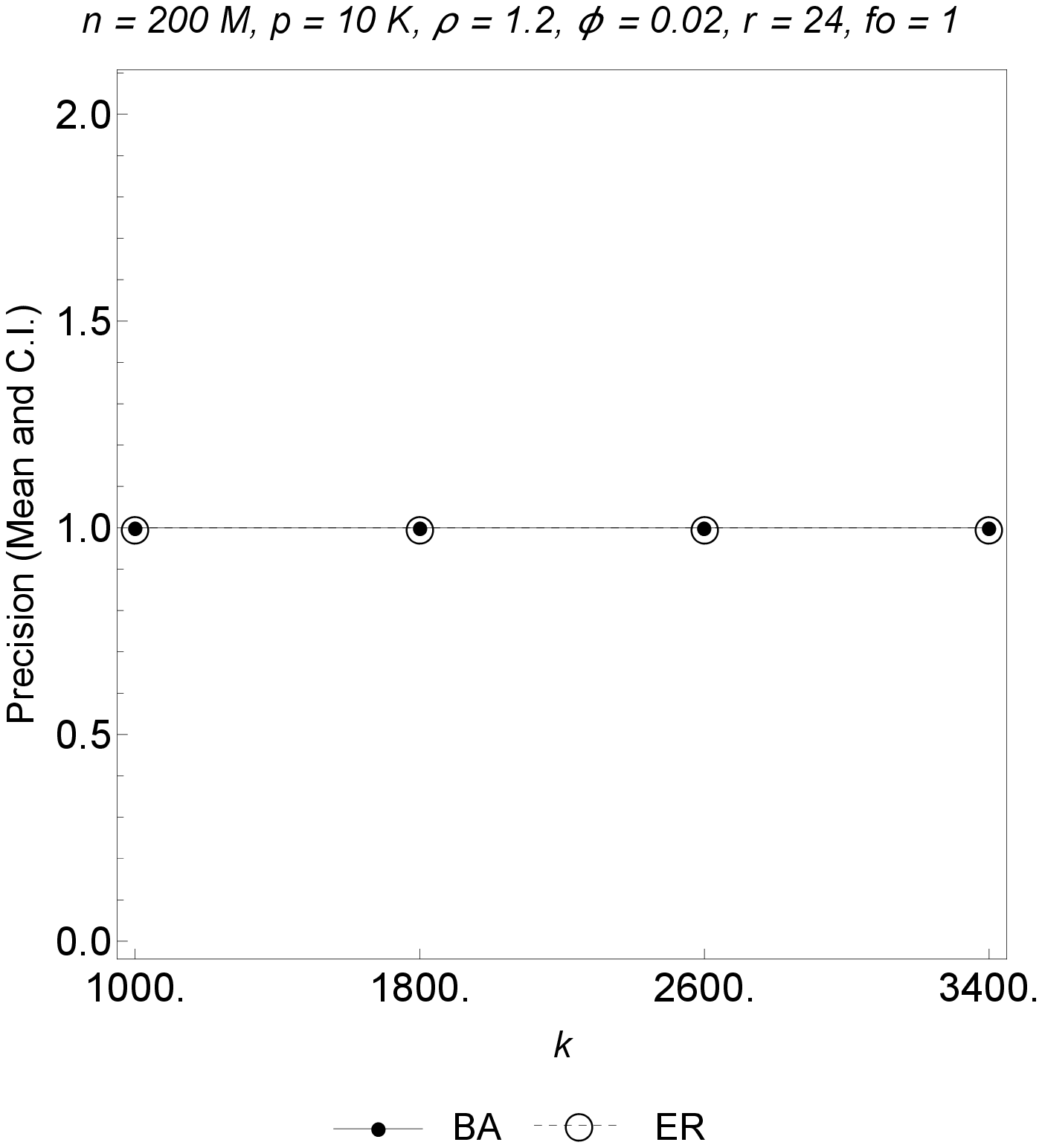}
			\label{k-prec}
		} &
		
			\subfloat[Average Relative Error ]{
				\includegraphics[width=0.3\textwidth]{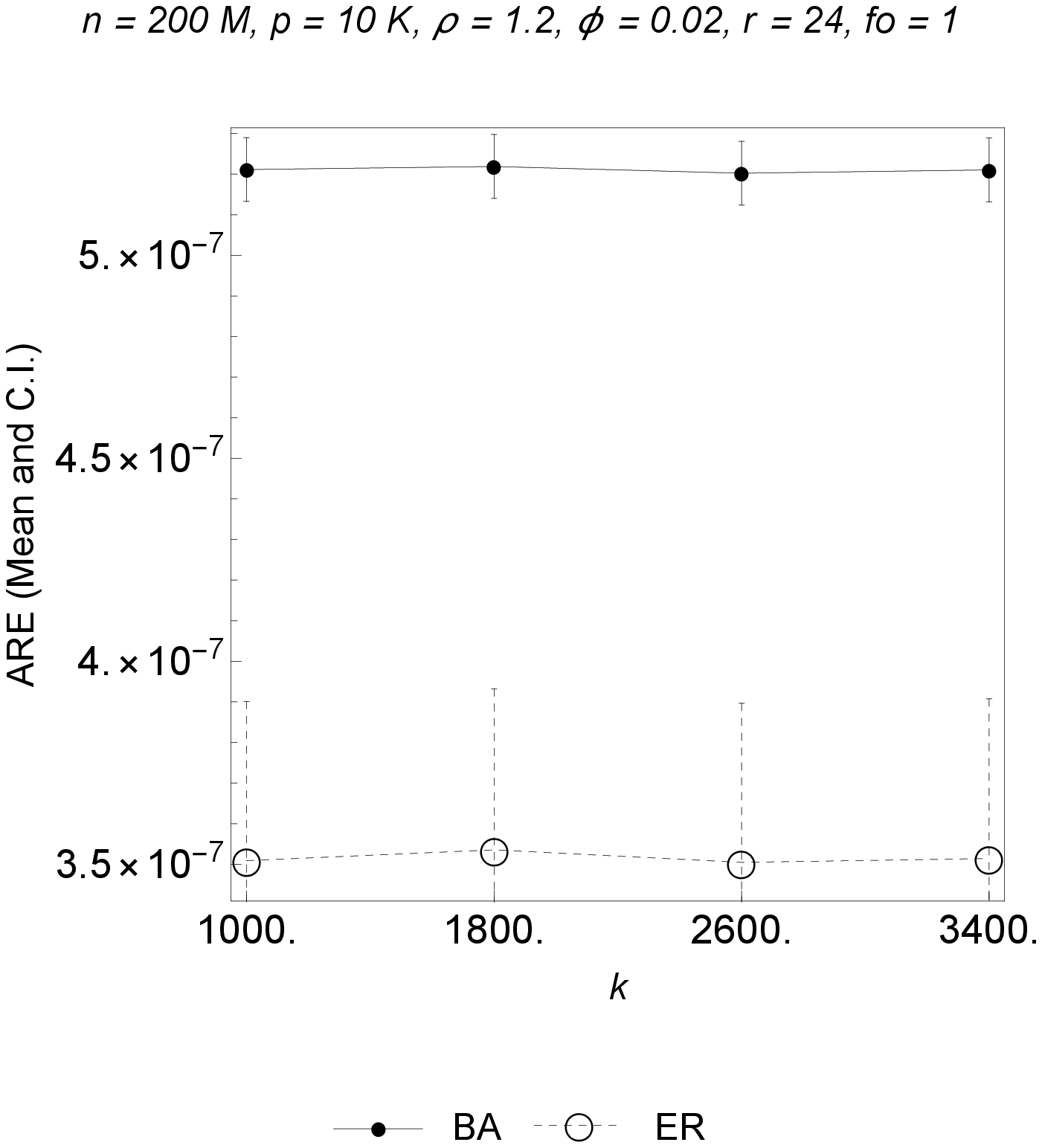}
				\label{k-are}
			} 
	\end{tabular}
	
	\caption{Recall, Precision and Average Relative Error (mean and confidence interval) varying the number of Space-Saving counters used by each peer,  for both a Barabasi-Albert (BA) and an Erdos-Renyi (ER) type of network graph.} 
	\label{k_plot}
\end{figure*}

The plots related to the experiments in which we varied the number $k$ of Space-Saving counters (Figure~\ref{k_plot}) do not present particular behaviours in the interval of values tested, showing that in this case the number of counters used were always enough with regard to the number of rounds executed in order to guarantee a good accuracy.

\begin{figure*}[h]
	\centering
	\begin{tabular}{ccc}		
		\subfloat[Recall]{
			\includegraphics[width=0.3\textwidth]{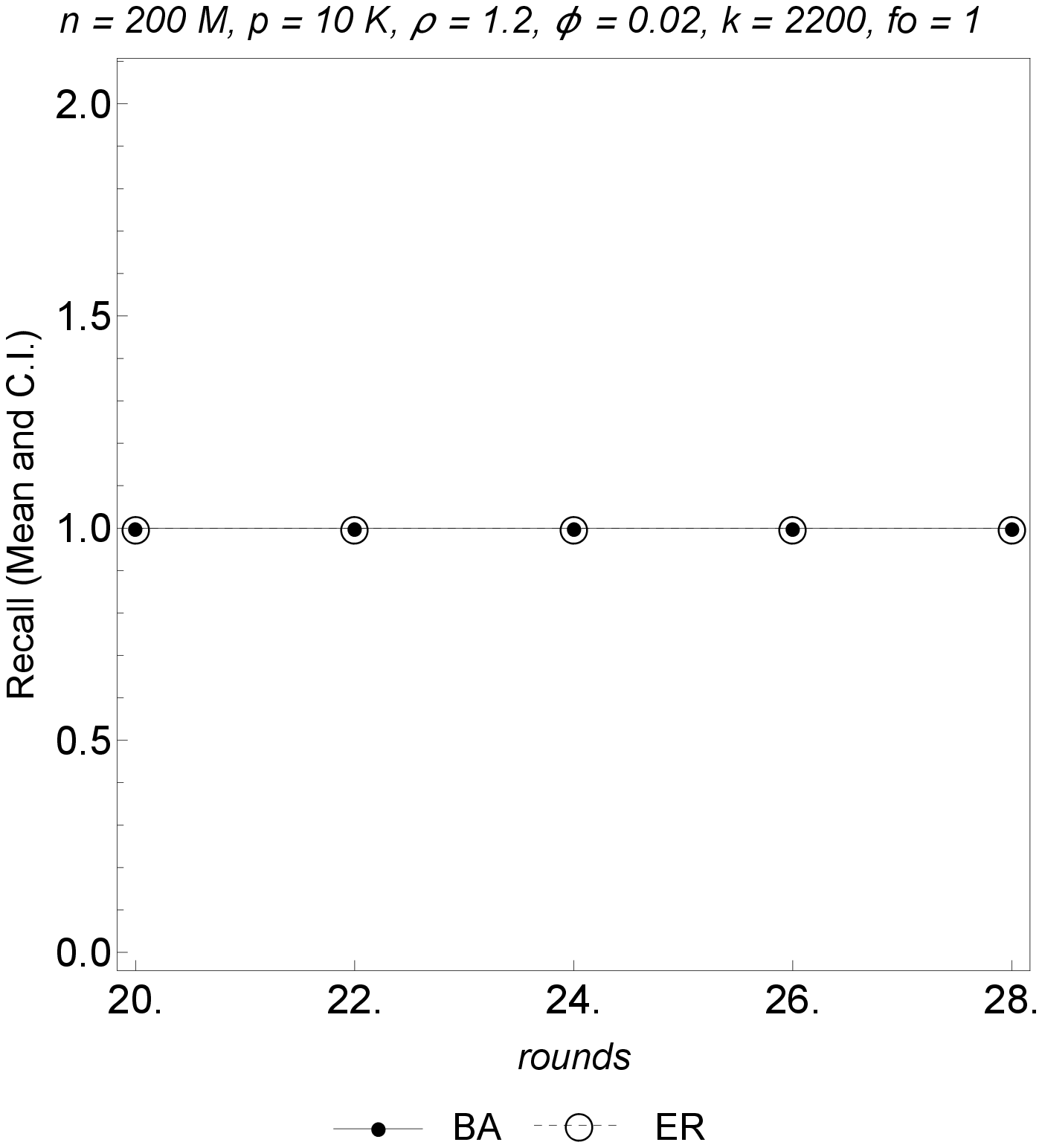}
			\label{r-rec}
		} &
		
		\subfloat[Precision]{
			\includegraphics[width=0.3\textwidth]{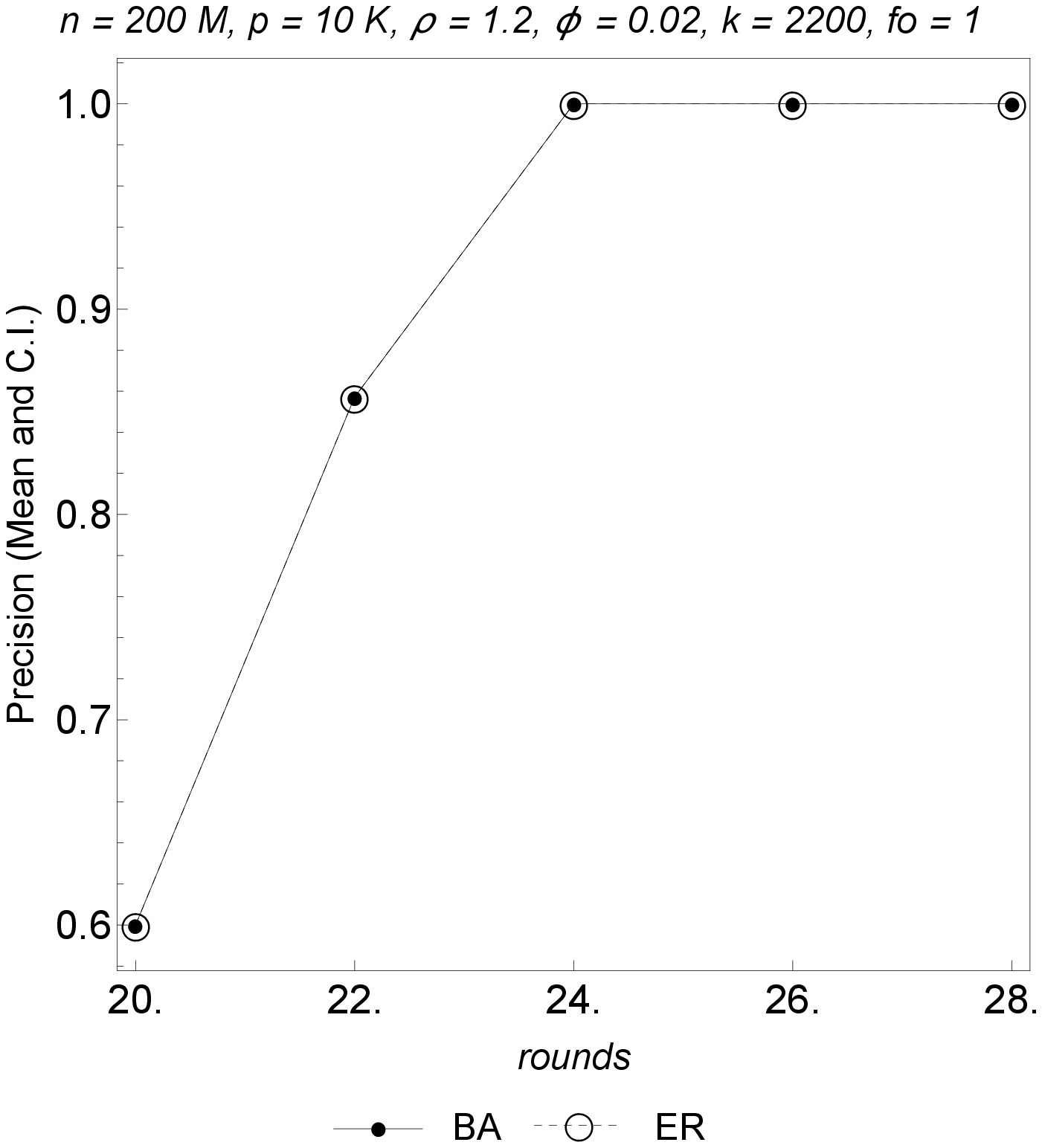}
			\label{r-prec}
		} &
		
			\subfloat[Average Relative Error ]{
				\includegraphics[width=0.3\textwidth]{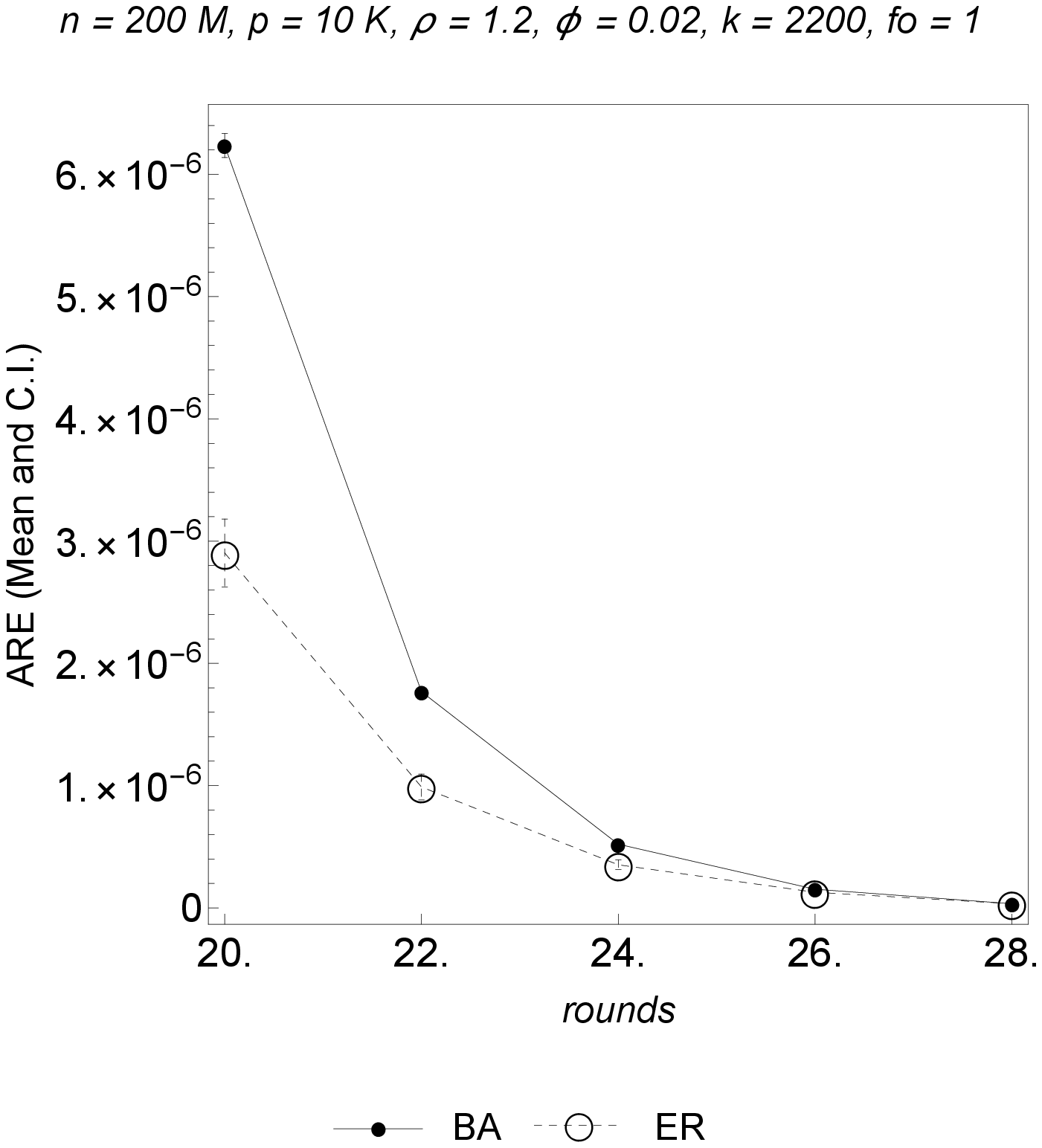}
				\label{r-are}
			} 

	\end{tabular}
	
	\caption{Recall, Precision and Average Relative Error (mean and confidence interval) varying the number of rounds executed,  for both a Barabasi-Albert (BA) and an Erdos-Renyi (ER) type of network graph.} 
	\label{r_plot}
\end{figure*}

A major sensitivity is exhibited by the algorithm when the number of rounds executed is varied, Figure~\ref{r_plot}. We note that the Precision grows and the Average Relative Error decreases as the number of rounds increases. This behaviour is expected, given the theoretical analysis. 

Overall the experiments show that our algorithm exhibits very good performance in terms of Recall, Precision, and Average Relative Error of the frequency estimation when the guidance of the theoretical analysis is taken into account in determining the number of counters used and the number of rounds to be executed. Furthermore, the algorithm proves to be very robust to variations in the skewness of the input dataset and the frequent items threshold.

\subsection{Effect of churn}
\label{effect_of_churn}
In order to verify the efficiency of our P2PSS algorithm in realistic P2P networks, we have carried out further experiments  introducing churning based on two different models: the \textit{fail-stop} model and the \textit{Yao} model, proposed by Yao et al. \cite{4110276}. 

In the \textit{fail-stop} model, a peer could leave the network with a given failure probability and the failed peers can not join the network anymore. 

In the \textit{Yao} model, peers randomly join and leave the network. For each peer $i$, a random average \textit{lifetime} duration $l_i$ is generated from a Shifted Pareto distribution with parameters $\alpha=3$, $\beta = 1$ and $\mu = 1.01$. Similarly, a random average \textit{offline} duration $d_i$ is generated from a Shifted Pareto distribution with the same $\alpha$ and $\mu$ parameter values and with $\beta = 2$. We recall here that if $X \sim \text{Pareto(II)}(\mu,\beta,\alpha)$, i.e., $X$ is a random variable with a Pareto Type II distribution (also named Shifted Pareto), then its cumulative distribution function is $F_X(x) = 1 - \left(1 + \frac{x - \mu}{\beta}\right)^{-\alpha}$.

The values $l_i$ and $d_i$  are used to configure, for each peer $i$, two distributions $F_i$ and $G_i$. The distributions $G_i$ are Shifted Pareto distributions with $\beta = 3$ and $\alpha = 2d_i$, whilst the distributions $F_i$ can be both Pareto distributions with $\beta = 2$, $\alpha = 2l_i$, or exponential distributions with $\lambda = 1/l_i$. Whenever the state of a peer changes, a duration value is drawn from one of the distributions, $F_i$ or $G_i$, based on the type of duration values (lifetime or offline) which must be generated.  We carried out our experiments with both the Pareto and Exponential lifetimes variants.

In order to correctly manage the churning of the peers, the algorithm must be modified as follows. We assume that a peer can detect a neighbour failure, then:
\begin{itemize}
	\item if a peer fails before sending a push message or after receiving a pull message, that is, when no  communications are ongoing, then no actions have to be performed;
	\item if a peer $p$ fails before sending a pull message to peer $r$ in response to its push message, then the peer $r$ detects the failure and simply cancels the push--pull exchange, so that its state does not change;
	\item if a peer $p$ fails after sending a push message to a peer $r$ and before receiving the corresponding pull message, then the peer $r$ detects the failure and restores its own local state as it was before the push--pull exchange.
\end{itemize}

When using the fail-stop model, we tested our algorithm with the default parameter values of Table \ref{experiments}, varying the failure probability through the values: $0.0, 0.01, 0.05, 0.1$. 
As shown in Figures~\ref{fp-rec} and \ref{fp-prec}, the recall and precision metrics are not affected at all by the introduction of peer failures up to a failure probability equal to $0.1$. However, as expected, Figure~\ref{fp-are} shows that the average relative error on frequency estimations gets worse going from about $10^{-6}$ in case of no churn to about $10^{-2}$ when the failure probability is $0.1$. 

\begin{figure*}[h]
	\centering
	\begin{tabular}{ccc}		
		\subfloat[Recall]{
			\includegraphics[width=0.3\textwidth]{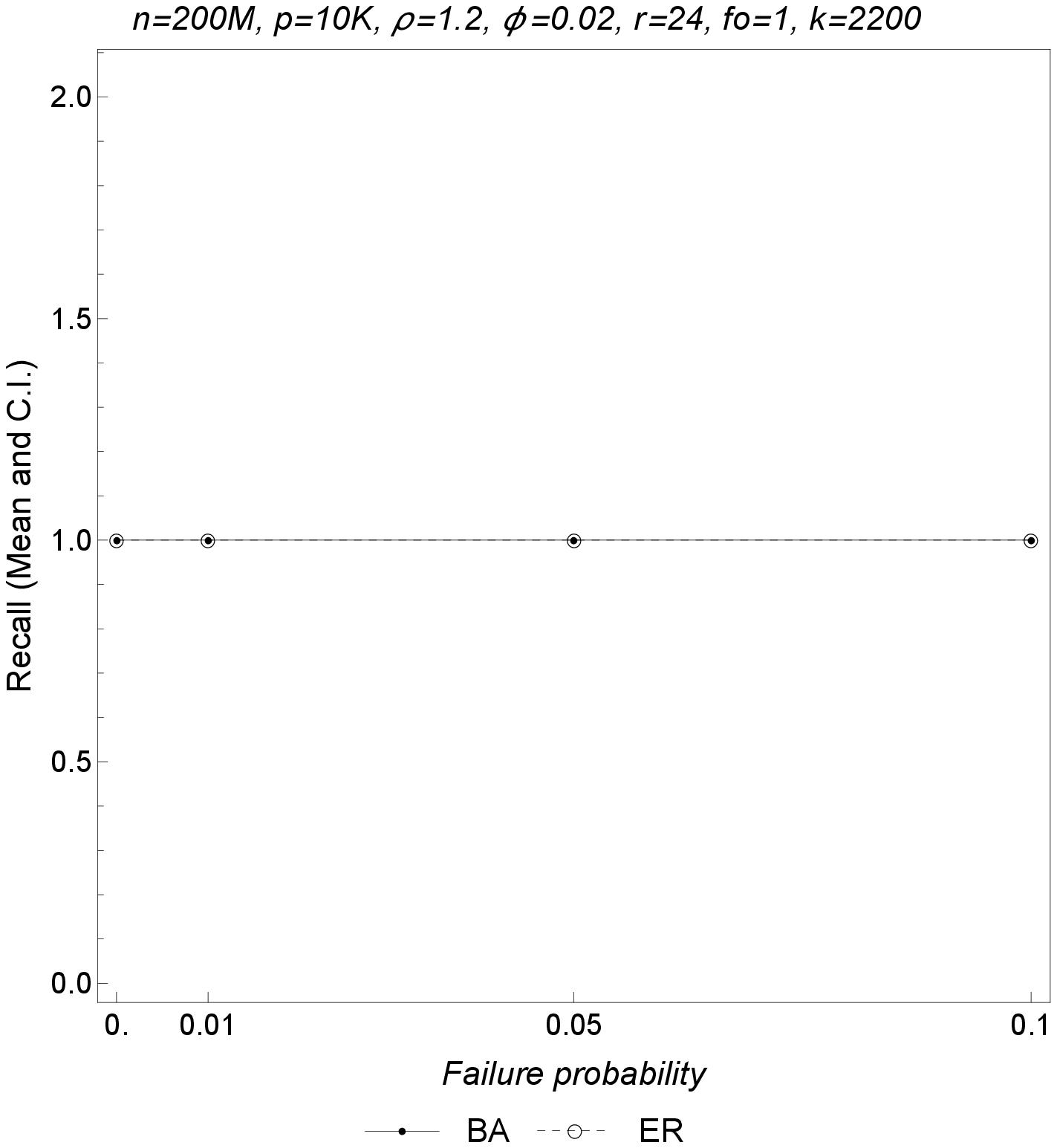}
			\label{fp-rec}
		} &
		
		\subfloat[Precision]{
			\includegraphics[width=0.3\textwidth]{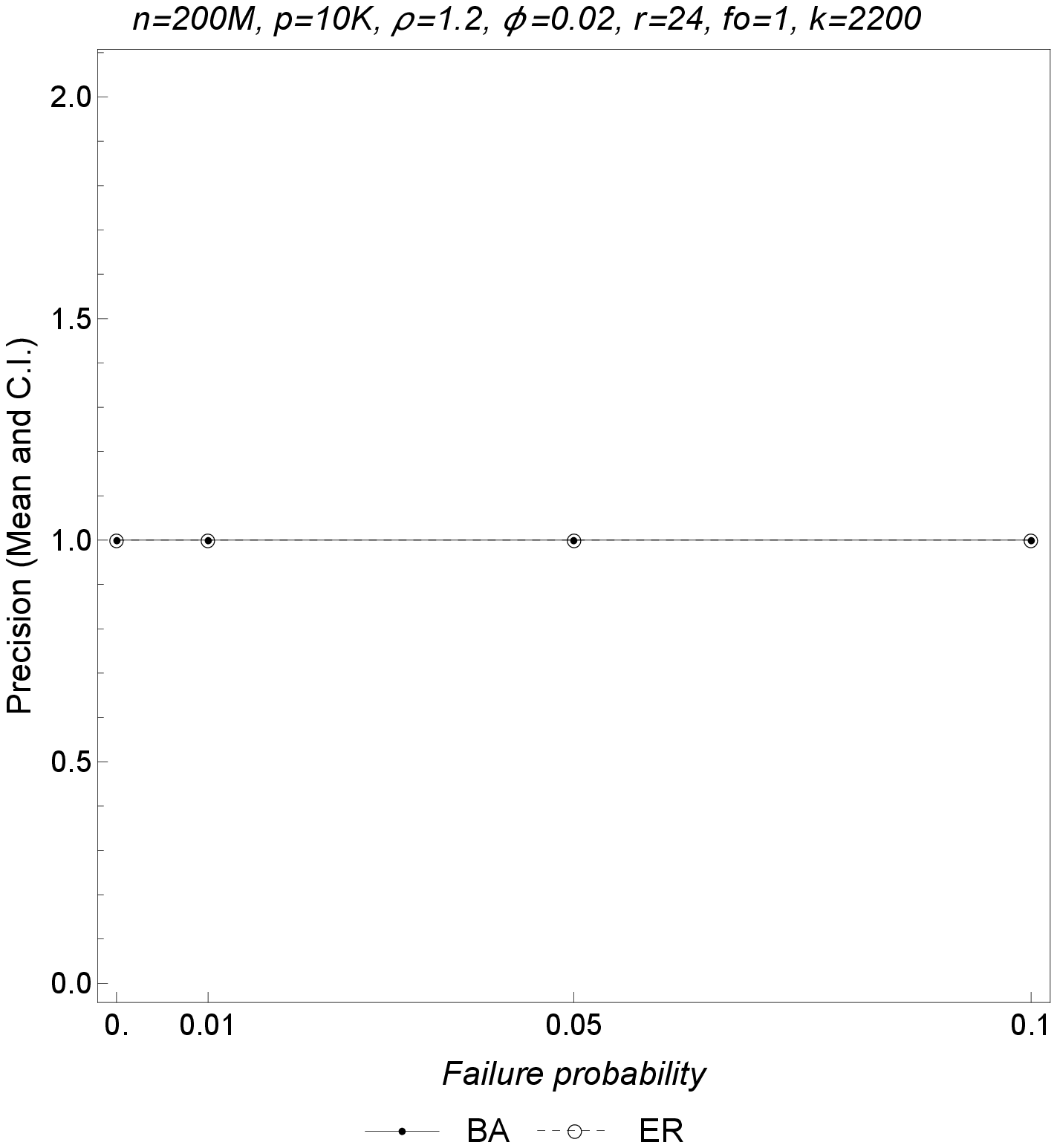}
			\label{fp-prec}
		} &
		
		\subfloat[Average Relative Error ]{
			\includegraphics[width=0.3\textwidth]{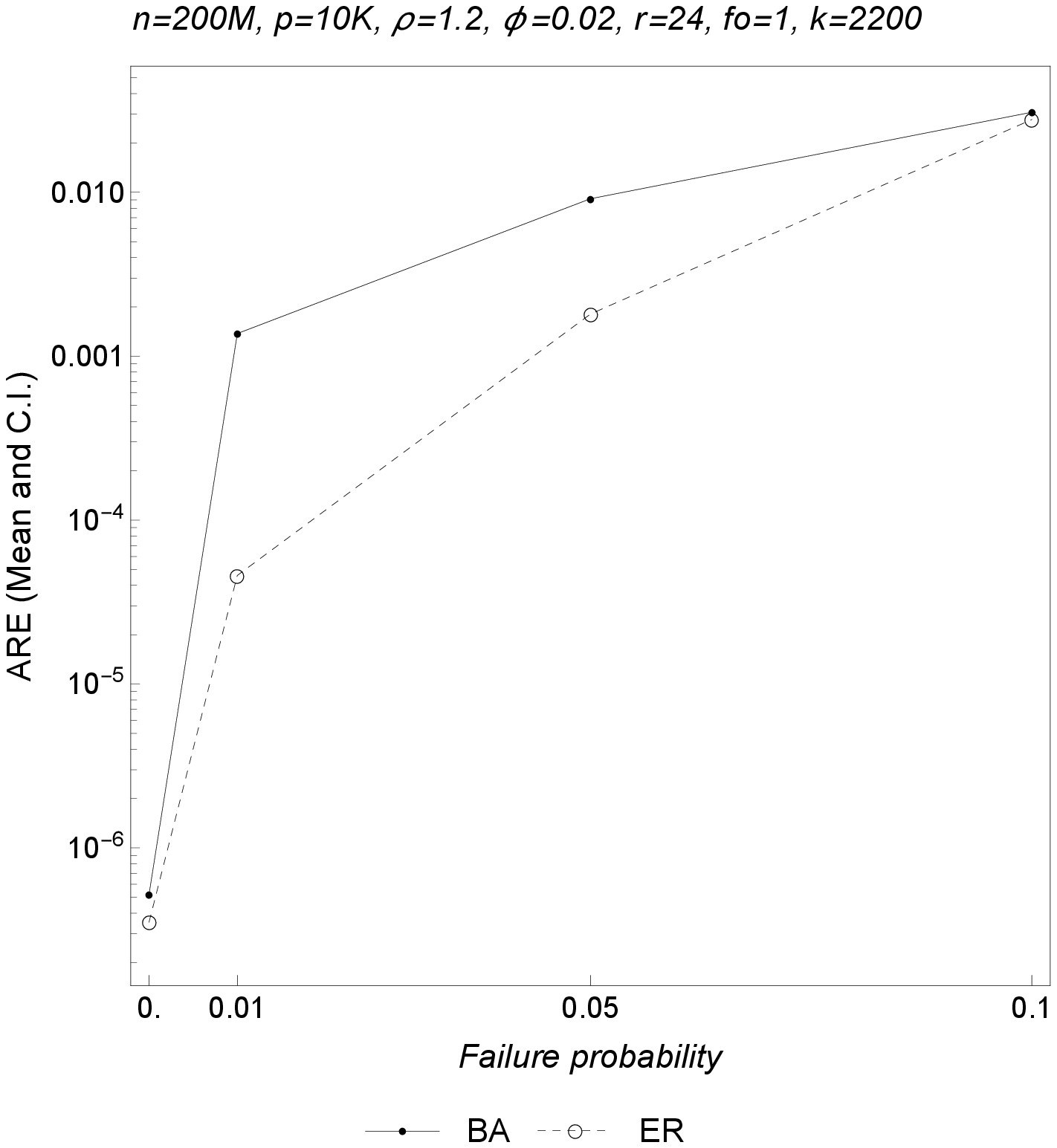}
			\label{fp-are}
		} 
		
	\end{tabular}
	
	\caption{Recall, Precision and Average Relative Error (mean and confidence interval) varying the failure probability in a fail-stop model of churning,  for both Barabasi-Albert (BA) and Erdos-Renyi (ER) random network graphs.} 
	\label{r_plot}
\end{figure*}

When the Yao model of churning was adopted, we tested our algorithm with the default values of Table \ref{experiments} and the parameters of churning already discussed, varying the maximum number of peers  and the number of rounds. Also in these cases, recall and precision are not  affected by the introduction of churning. Indeed, we obtained for recall and precision varying the number of peers and the number of rounds the same plots as Figures~\ref{p-rec}, \ref{p-prec}, \ref{r-rec} and \ref{r-prec}; for this reason we do not report these plots again. On the other hand, the average relative error is affected by the churning, as expected: Figures \ref{p-churn} and \ref{r-churn} are related respectively to the ARE measured varying the number of peers and the number of rounds with Pareto distributions for lifetimes, whilst Figures \ref{p-expchurn} and \ref{r-expchurn} refer to the ARE measured when using Exponential distributions for lifetimes.

\begin{figure*}[h]
	\centering
	\begin{tabular}{cc}		
		\subfloat[Average Relative Error]{
			\includegraphics[width=0.45\textwidth]{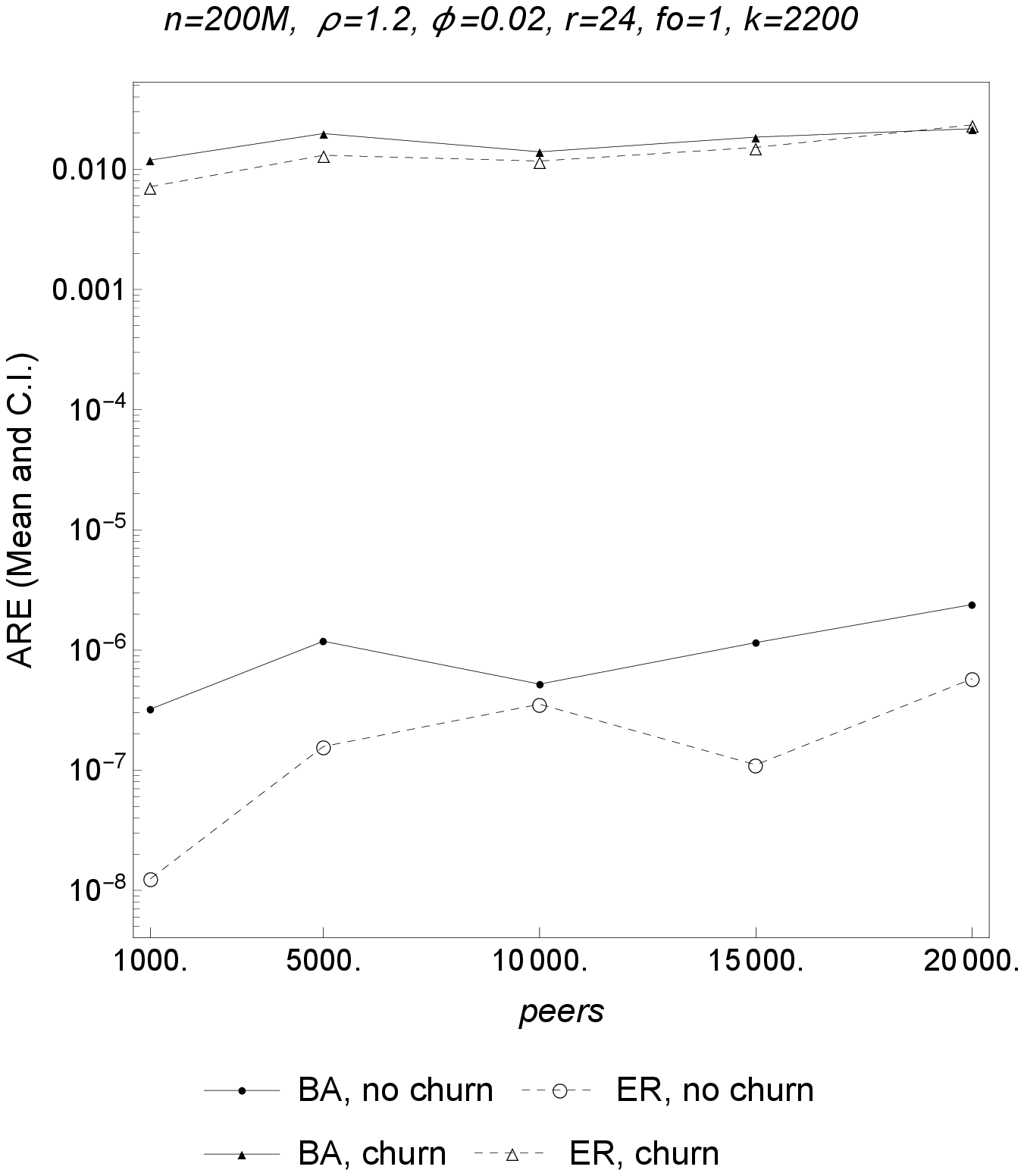}
			\label{p-churn}
		} &
		
		\subfloat[Average Relative Error]{
			\includegraphics[width=0.45\textwidth]{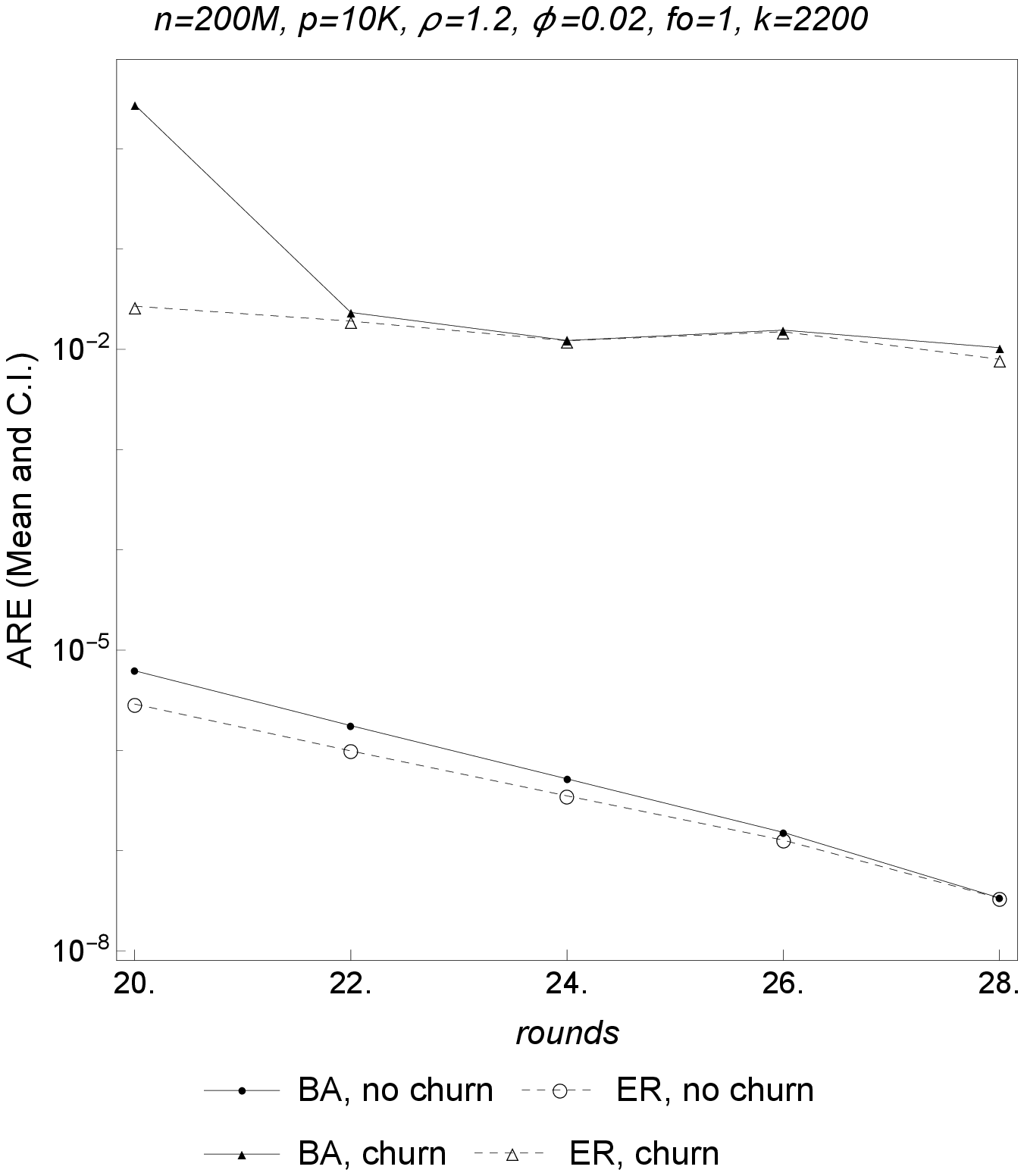}
			\label{r-churn}
		} \\
		
		\subfloat[Average Relative Error ]{
			\includegraphics[width=0.45\textwidth]{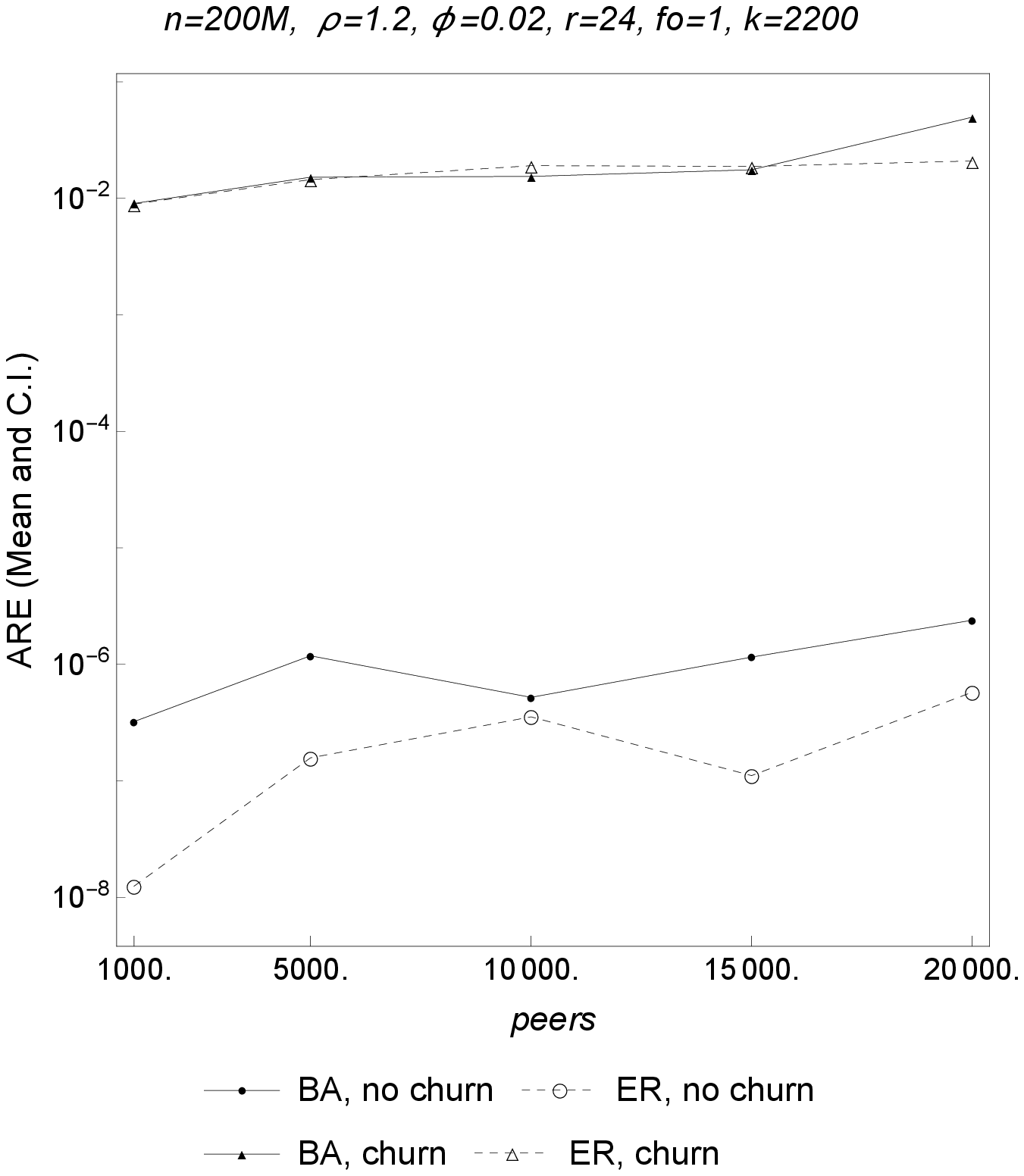}
			\label{p-expchurn}
		} &
	
		\subfloat[Average Relative Error ]{
			\includegraphics[width=0.45\textwidth]{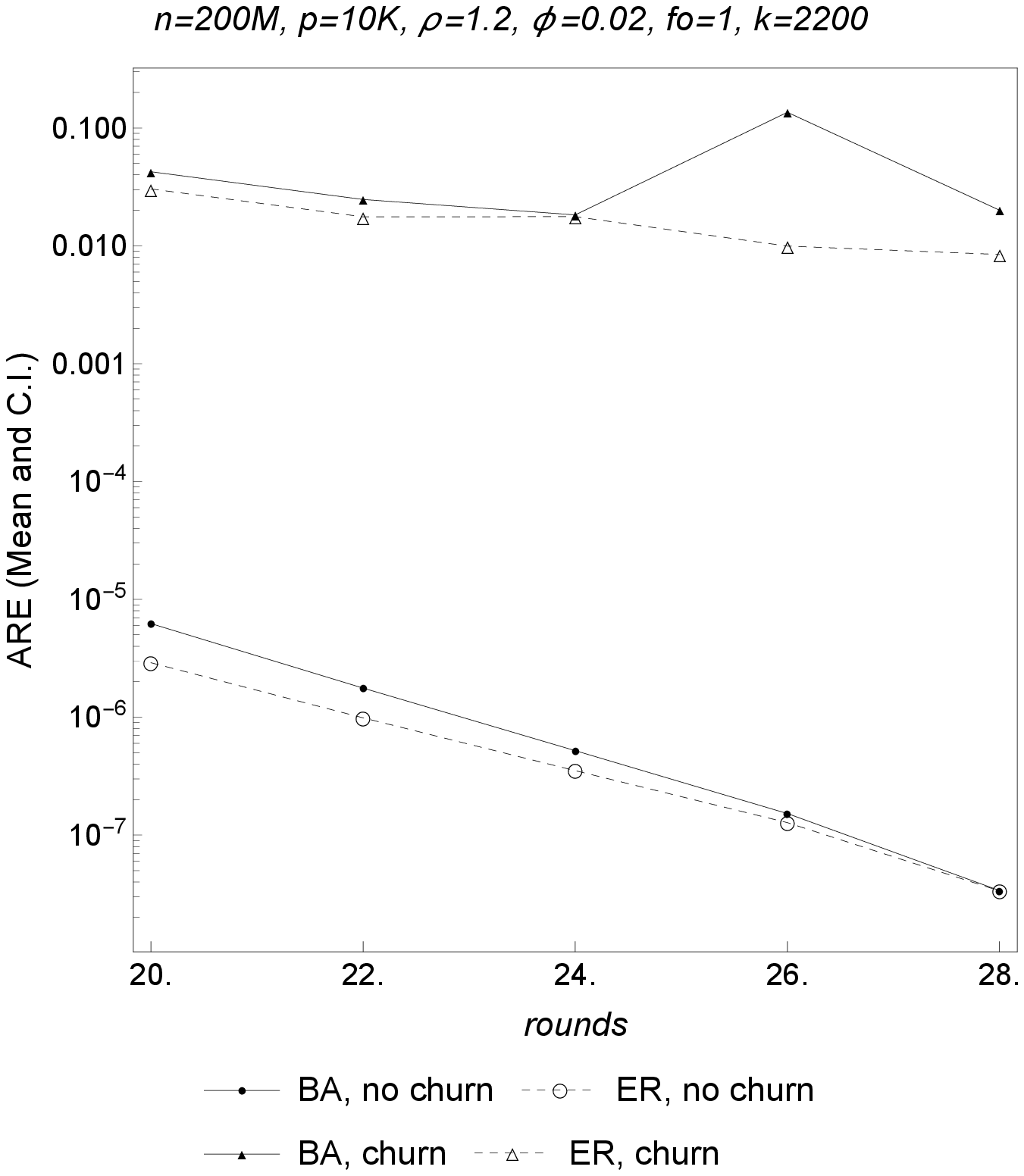}
			\label{r-expchurn}
		} 
		
	\end{tabular}
	
	\caption{Average Relative Error (mean and confidence interval) varying the number of peers and rounds in the Yao model of churning with Pareto (a,b) or Exponential (c, d) lifetimes,  for both Barabasi-Albert (BA) and Erdos-Renyi (ER) random network graphs.} 
	\label{r_plot}
\end{figure*}

\section{Related work}
\label{related}

In this Section, we recall the most important sequential, parallel and distributed algorithms for the frequent items problem. The items to be mined may belong to either a static dataset or, in the most general setting, to a stream. In the former case, all of the data is already available in advance, whilst in the latter data arrives or can be accessed only sequentially and in a given order; no random access to the data is allowed. 

Sequential algorithms can be broadly classified as either deterministic, counter--based or randomized, sketch--based. A counter--based algorithm works by updating a so called summary (or synopsis) data structure. The summary is updated at each item arrival and requires a bounded amount of memory, much smaller than that necessary for storing the entire input. Queries are answered using that summary, and the time for processing an item and computing the answer to a given query is limited. Sketch--based algorithms process items using a sketch, which is a bi-dimensional array of counters. Each input item is mapped, through hash functions, to corresponding sketch cells whose values are then updated as required by the algorithm. 

The seminal counter--based algorithm proposed by Misra and Gries \cite{Misra82} has been independently rediscovered and improved (with regard to its computational complexity) by Demaine et al. \cite{DemaineLM02} (the so-called \emph{Frequent} algorithm) and Karp et al. \cite{Karp}. Among the counter--based algorithms, we recall here \emph{Sticky Sampling}, \emph{Lossy Counting} \cite{Manku02approximatefrequency}, and \emph{Space-Saving} \cite{Metwally2006}. In Particular, among counter--based algorithms, Space-Saving provides the best accuracy whilst requiring the minimum number of counters and constant time complexity to update its summary upon an item arrival. Notable sketch--based algorithms are \emph{CountSketch} \cite{Charikar}, \emph{Group Test} \cite{Cormode-grouptest}, \emph{Count-Min} \cite{Cormode05} and \emph{hCount} \cite{Jin03}.

Regarding parallel algorithms, \cite{cafaro-tempesta} (slightly improved in \cite{cafaro-pulimeno}) and \cite{cafaro-pulimeno-tempesta} present message-passing based parallel versions of the Frequent and Space-Saving algorithms. Among the algorithms for shared-memory architectures we recall here a parallel version of Frequent \cite{Zhang2013}, a parallel version of Lossy Counting \cite{Zhang2012}, and parallel versions of Space-Saving \cite{Roy2012} and \cite{Das2009}. Novel shared-memory parallel algorithms for frequent items were recently proposed in \cite{Tangwongsan2014}. Accelerator based algorithms for frequent items exploiting a GPU (Graphics Processing Unit) include \cite{Govindaraju2005}, \cite{Erra2012}, \cite{cafaro-epicoco-aloisio-pulimeno} and \cite{CPE:CPE4160}. 

Some applications are concerned with the problem of detecting frequent items in a stream with the additional constraint that recent items must be weighted more than former items. The underlying assumption is that recent data is certainly more useful and valuable than older, stale data. Therefore, each item in the stream has an associated timestamp that shall be used to determine its weight. In practice, instead of estimating frequency counts, an application must be able to estimate \emph{decayed counts}. Two different models have been proposed in the literature: the \emph{sliding window} and the \emph{time fading} model.

In the sliding window model \cite{Datar} \cite{TCS-002}, freshness of recent items is captured by a time window, i.e., a temporal interval of fixed size in which only the most recent $N$ items are taken into account; detection of frequent items is strictly related to those items falling in the window. The items in the stream become stale over time, since the window periodically slides forward.

The time fading model \cite{exp-decay} does not use a window sliding over time; freshness of more recent items is instead emphasized by \emph{fading} the frequency count of older items. This is achieved by using a decaying factor $0 < \lambda < 1$ to compute an item's \textit{decayed count} (also called \textit{decayed frequency}) through decay functions that assign greater weight to more recent elements. The older an item, the lower its decayed count is: in the case of exponential decay, the weight of an item occurred $n$ time units in the past, is $e^{-\lambda n}$, which is an exponentially decreasing quantity. Mining time faded frequent items has been investigated in \cite{Chen-Mei}, \cite{Cafaro-Pulimeno-Epicoco-Aloisio}, \cite{Wu2017}, \cite{8091102}. A parallel message-passing based algorithm has been recently proposed in \cite{CAFARO2018115}. 

Regarding the Correlated Heavy Hitters Problem (CHHs), an algorithm based on the nested application of Frequent has been recently presented in \cite{Lahiri2016}. The outermost application mines the primary dimension, whilst the innermost one mines correlated secondary items. The main drawbacks of this algorithm, being based on Frequent, are the accuracy (which is very low), the huge amount of space required and the rather slow speed (owing to the nested summaries).

In \cite{Epicoco:2018:FAM:3182040.3182103}, a faster and more accurate algorithm for mining CHHs is proposed. The Cascading Space-Saving Correlated Heavy Hitters (CSSCHH) algorithm exploits the basic ideas of Space-Saving, combining two summaries for tracking the primary item frequencies and the tuple frequencies. The algorithm is referred to as Cascading Space-Saving since it is based on the use of two distinct and independent applications of Space-Saving.

Let us now discuss related work focusing on the P2P approach. Since our algorithm is designed for unstructured P2P networks and is based on a gossip protocol \cite{Demers:1987}, among the many distributed algorithms for mining frequent items (e.g., \cite{Cao:2004}, \cite{Zhao:2006}, \cite{Keralapura:2006}, \cite{recent-freq-items}, \cite{Venkataraman}, \cite{MENG201629}) we only discuss \cite{computing13}, \cite{CEM20131544}, \cite{LAHIRI20101241}. 

The algorithms presented in \cite{computing13} and \cite{CEM20131544} are very similar. Each peer starts with a local subset of the whole dataset to be mined, and it is explicitly assumed that each peer can store the whole dataset, i.e., the dataset resulting from the union of the local datasets; the whole dataset is obtained as a result of the periodic gossip exchanges, in which the peers send their local dataset, receive their neighbours' datasets and merge them; this is known as averaging gossip protocol. The number of peers can be estimated by using the same approach, in which one of the peers starts with a value equal to one and all of the others with a value equal to zero. The convergence properties of the averaging gossip protocol have been thoroughly studied in \cite{Jelasity2005}, and it has been shown that each round contributes to reducing the variance around the mean value that is being computed.

In order to reduce the communication complexity, \cite{computing13} suggests alternatively to exchange only the top-$k$ most frequent items where $k$ is a user's defined parameter. The termination condition is based on the following convergence criterion: the algorithm stops when for all of the peers, the subset consisting of the top-$k$ items does not change for a specified number of consecutive rounds.  

The algorithm presented in \cite{CEM20131544} tries to reduce the communication complexity in a different way. It uses an additional data structure, an hash table in which all the items seen are stored (these items are never deleted) and from which the algorithm randomly selects a specified number of items corresponding to a predefined message size. The termination condition is based on a convergence criterion requiring two user's defined parameters: $\epsilon$ and $convLimit$. If the absolute difference between the true and the estimated frequencies of all of the itemss is less than or equal to $\epsilon$ for at least $convLimit$ consecutive rounds, the algorithm stops its execution.

It is clear from the previous discussion that \cite{computing13} and \cite{CEM20131544} require space complexity linear in the length $n$ of the dataset; this allows solving the \textit{exact} problem rather than the approximate problem. 

In \cite{LAHIRI20101241}, the authors provide a randomized approach based on a random sampling of the items and the averaging gossip protocol. A random weight in the interval (0, 1) is assigned to each item. The algorithm maintains and exchanges in each round a data structure consisting of $t$ items whose weight is the lowest, where $t = \frac{128}{\psi^2} \ln \frac{3}{\delta}$, $\psi$ is an error threshold and $\delta$ the probability of failure. Even though the authors prove the theoretical properties of their algorithm, we remark here that the approach can only detect frequent items but does not provide any kind of frequency estimation: the algorithm returns a list of items that with high probability (defined by $\delta$) contains the frequent items (with regard to the $\psi$ threshold). Regarding the space used, for each of the $t$ items the algorithm stores a tuple consisting of four fields: the peer identifier, the item index in the peer's local dataset, the item value and its random weight.

We remark here that \cite{computing13}, \cite{CEM20131544} do not solve the \textit{Approximate Frequent Items Problem in Unstructured P2P Networks} and that \cite{LAHIRI20101241} does not provide frequency estimation of the discovered frequent items. In contrast, our algorithm solves the Approximate Frequent Items Problem in Unstructured P2P Networks, and it does so by using very little space: each peer uses exactly the same stream summary data structure that would be used by a centralized algorithm. Moreover, to the best of our knowledge, we provide the first distributed algorithm for the Approximate Frequent Items Problem in Unstructured P2P Networks using a gossip--based protocol with strong theoretical guarantees for both the Approximate Frequent Items Problem in Unstructured P2P Networks and for frequency estimation of the discovered frequent items.

\section{Conclusions}
\label{conclusions}

In this paper, we have dealt with the problem of mining frequent items in unstructured P2P networks. This problem, of practical importance, has many useful applications. We have designed P2PSS, a fully decentralized, gossip--based protocol for frequent items discovery, leveraging the Space-Saving algorithm. We have formally proved the correctness and theoretical error bound of the algorithm, and shown, through extensive experimental results, that P2PSS provides very good accuracy and scalability, also in the presence of highly dynamic P2P networks with churning. To the best of our knowledge, this is the first gossip--based distributed algorithm providing strong theoretical guarantees for both the Approximate Frequent Items Problem in Unstructured P2P Networks and for frequency estimation of the discovered frequent items.

\clearpage

\bibliographystyle{elsarticle-num}
\bibliography{bibliography}

\end{document}